\documentclass[letterpaper,USenglish,11pt]{article}

\usepackage{comment}
\usepackage{ifthen}
\usepackage{xspace}
\usepackage{float}
\usepackage{subcaption}
\usepackage{amsmath}
\usepackage{amsthm}
\usepackage{amstext}
\usepackage{amssymb}
\usepackage{amsfonts}
\usepackage{bm}
\usepackage[multiple]{footmisc}
\usepackage{wrapfig}
\usepackage{url}
\usepackage[margin=1in]{geometry}
\usepackage{mathdots}
\usepackage{mathtools}
\usepackage{thmtools}
\usepackage{thm-restate}
\usepackage{tikz}
\usetikzlibrary{calc, fit, shapes,intersections,arrows,matrix, topaths}
\usetikzlibrary{decorations,decorations.pathreplacing,patterns, shadows}
\usepackage{graphicx}
\usepackage[linesnumbered,noresetcount,vlined,ruled]{algorithm2e}
\usepackage[inline]{enumitem}
\usepackage{environ}
\usepackage{dsfont}
\usepackage{cite}
\usepackage{yfonts} % not realy needed - used for crazy font

\newcommand{\FP}{\text{FP}}
\newenvironment{proofsketch}{%
  \proof}{\endproof}
\renewcommand{\epsilon}{\varepsilon} 

\renewcommand{\epsilon}{\varepsilon}
\newcommand{\eps}{\varepsilon}
\newcommand{\Dict}{\textsf{Dict}}

\usepackage{amsmath} \usepackage{amsthm}
\newtheorem{theorem}{Theorem}%[section]
\newtheorem{lemma}[theorem]{Lemma}

\newtheorem{claim}[theorem]{Claim}

\newtheorem{assumption}[theorem]{Assumption}
\newtheorem{definition}[theorem]{Definition}
\newtheorem{invariant}[theorem]{Invariant}
\newtheorem{observation}[theorem]{Observation}

\newcommand{\UU}{\mathcal{U}}
\newcommand{\MM}{\mathcal{M}}

\newcommand{\hc}{h^c}
\newcommand{\qsi}{q^{si}}

\newcommand{\hb}{h^b}
\newcommand{\hf}{\hat{f}}

\newcommand{\hatc}{\hat{c}}
\newcommand{\hl}{\hat{\ell}}
\newcommand{\tf}{\tilde{f}}
\newcommand{\tl}{\tilde{\ell}}
\newcommand{\thash}{\tilde{h}}
\newcommand{\bh}{\bar{h}}
\newcommand{\mumacro}{\sqrt{
      \frac{3}{m}\cdot\ln(6w^{\beta/3})}}
\newcommand{\DD}{\mathcal{D}}
\newcommand{\etal}{\textit{et al.}\xspace}
\newcommand{\size}[1]{\ensuremath{\left|#1\right|}}
\newcommand{\set}[1]{\left\{ #1 \right\}}
\newcommand{\parentheses}[1]{\left(#1\right)}
\DeclarePairedDelimiter{\floor}{\lfloor}{\rfloor}
\DeclarePairedDelimiter{\ceil}{\lceil}{\rceil}

\renewcommand{\Pr}[1]{{\mathrm{Pr}}\left[ #1 \right]}
\newcommand{\NN}{\mathbb{N}}
\DeclareMathOperator{\query}{\textsf{query}}

\DeclareMathOperator{\ins}{\textsf{insert}}
\DeclareMathOperator{\pop}{\textsf{pop}}
\DeclareMathOperator{\del}{\textsf{delete}}

\DeclareMathOperator{\op}{\textsf{op}}
\DeclareMathOperator{\out}{\textsf{out}}

\DeclareMathOperator{\hheader}{\textsf{header}}
\DeclareMathOperator{\bbody}{\textsf{body}}
\DeclareMathOperator{\PD}{\textsf{PD}}
\DeclareMathOperator{\CSD}{\textsf{CSD}}
\DeclareMathOperator{\VarPD}{\textsf{VarPD}}
\DeclareMathOperator{\VarCSD}{\textsf{VarCSD}}

\DeclareMathOperator{\SID}{\textsf{SID}}
\DeclareMathOperator{\ptr}{\textsf{ptr}}

\DeclareMathOperator{\poly}{\textsf{poly}}
\newcommand{\ind}[1]{\mathds{1}_{#1}}
\newcommand{\indD}{\ind{\DD(t)}}
\makeatother
\newcommand{\betaval}{6+\delta}
\newcommand{\betapval}{5+0.75\cdot \delta}
\newcommand{\CC}{C}

\begin{document}
\title{Fully-Dynamic Space-Efficient Dictionaries and Filters with Constant
Number of Memory Accesses\thanks{This research was supported by a
grant from the United States-Israel Binational Science Foundation
(BSF), Jerusalem, Israel, and the United States National Science
Foundation (NSF)}}
\author{
 Ioana O. Bercea\thanks{
Tel Aviv University, Tel Aviv, Israel.
Email:~\texttt{ioana@cs.umd.edu, guy@eng.tau.ac.il}.}
\and
Guy Even\footnotemark[1]
}
\date{\today}
\maketitle

\begin{abstract} 
A fully-dynamic dictionary is a data structure for maintaining sets
that supports insertions, deletions and membership queries. A filter
approximates membership queries with a one-sided error. We present
two designs:
\begin{enumerate}
\item The first space-efficient fully-dynamic dictionary that
maintains both sets and random multisets and supports
queries, insertions, and deletions with
a constant number of memory accesses in
the worst case with high probability. The comparable
dictionary of
Arbitman~\etal~\cite{arbitman2010backyard} works only for
sets.
\item By a reduction from our dictionary for random
multisets, we obtain a space-efficient fully-dynamic filter that
supports queries, insertions, and deletions with a constant
number of memory accesses in the worst case with high
probability (as long as the false positive probability is
$2^{-O(w)}$, where $w$ denotes the word length). This
is the first in-memory space-efficient fully-dynamic filter design
that provably achieves these properties. 
\end{enumerate}
We also present an application of the techniques used to design our
dictionary to the static Retrieval Problem.
\end{abstract}

\section{Introduction}
We consider the problem of maintaining datasets
subject to insert, delete, and membership query operations. 
Given a set $\DD$ of $n$ elements from a universe $\mathcal{U}$, a
membership query asks if the queried element $x \in \mathcal{U}$
belongs to the set $\DD$. When exact answers are required, the
associated data structure is called a \emph{dictionary}. When
one-sided errors are allowed, the associated data structure is called
a \emph{filter}. Formally, given an error parameter $\eps>0$, a filter
always answers ``yes'' when $x \in \DD$ and when $x\notin \DD$, it
makes a mistake with probability at most $\eps$. We refer to such an
error as a \emph{false positive} event%
\footnote{The probability is taken over the random choices that the
filter makes.}.

When false positives can be tolerated, the main advantage of using a
filter instead of a dictionary is that the filter requires much less
space than a dictionary~\cite{carter1978exact,lovett2010lower}.  Let
$u\triangleq |\UU|$ be the size of the universe and $n$ denote
an upper bound on the size of the set at all points in time. The information
theoretic lower bound for the space of dictionaries is
$\ceil{\log_2 {u \choose n}} = n\log(\frac{u}{n})+ \Theta(n)$
bits.\footnote{All logarithms are base $2$ unless otherwise
stated. $\ln x$ is used to denote the natural
logarithm.}\footnote{We consider the case in which the bound $n$ on
the cardinality of the dataset is known in advance. The scenario in
which $n$ is not known is addressed in Pagh~\etal~\cite{pagh2013approximate}.}\footnote{ This equality holds when $u$ is
significantly larger than $n$.} On the other hand, the lower bound
for the space of filters is $n\log(1/\eps)$ bits~\cite{carter1978exact}. In light of these lower bounds, we call a dictionary \textit{space-efficient} when
it requires $(1+o(1))\cdot n\log(\frac{u}{n})+ \Theta(n)$ bits, where the
term $o(1)$ converges to zero as $n$ tends to infinity. Similarly, a
space-efficient filter requires $(1+o(1))\cdot n\log(1/\eps) + O(n)$
bits.\footnote{An asymptotic expression that mixes big-O and small-o
calls for elaboration. If $\eps=o(1)$, then the asymptotic
expression does not require the $O(n)$ addend. If $\eps$ is constant, the
$O(n)$ addend only emphasizes the fact that the constant that
multiplies $n$ is, in fact, the sum of two constants: one is almost
$\log (1/\eps)$, and the other does not depend on $\eps$. Indeed,
the lower bound in~\cite{lovett2010lower} excludes space
$(1+o(1))\cdot n\log(1/\eps)$ for constant values of $\eps$.}

When the set $\DD$ is fixed, we say that the data structure is \emph{static}. When the data structure
supports insertions as well, we say that it is \emph{incremental}. Data structures
that handle both deletions and insertions are called \emph{fully-dynamic} (in short, dynamic). 

The goal is to design  dictionaries and filters
that achieve ``the best of both worlds''~\cite{arbitman2010backyard}:
they work in the fully-dynamic setting, are space-efficient and
perform operations in constant time in the worst case with high
probability.\footnote{By with high probability (whp), we mean with
  probability at least $1-1/n^{O(1)}$. The constant in the exponent
  can be controlled by the designer and only affects the $o(1)$ term
  in the space of the dictionary or the filter.}

On the fully-dynamic front, one successful approach for designing
fully-dynamic filters was suggested by
Pagh~\etal~\cite{DBLP:conf/soda/PaghPR05}.  Static (resp.,
incremental) filters can be obtained from static (resp., incremental)
dictionaries for sets by a reduction of
Carter~\etal~\cite{carter1978exact}. This reduction does not directly
lead to filters that support deletions.  To extend the reduction to
the fully-dynamic setting, Pagh~\etal~\cite{DBLP:conf/soda/PaghPR05}
propose to reduce instead from a stronger dictionary that maintains
multisets rather than sets (i.e., elements in multisets have arbitrary
multiplicities). This extension combined with a fully-dynamic dictionary for
multisets yields in~\cite{DBLP:conf/soda/PaghPR05} a fully-dynamic
filter that is space-efficient but performs insertions and deletions
in amortized constant time, not in the worst case.  It is still an
open problem whether one can design a fully-dynamic dictionary on
multisets that is space-efficient and performs operations in constant
time in the worst case whp~\cite{arbitman2010backyard}.

In this paper, we present the first fully-dynamic space-efficient
filter with constant time operations in the worst case whp.  Our
result is based on the observation that it suffices to use the
reduction of Carter~\etal~\cite{carter1978exact} on 
dictionaries that support \emph{random} multisets rather than
arbitrary multisets. We then design the first fully-dynamic
space-efficient dictionary that works on random multisets and supports
operations in constant time in the worst case whp.  Applying the
reduction to this new dictionary yields our fully-dynamic filter. We
also show how a static version of our dictionary can be used to design
a data structure for the static retrieval problem.

\subsection{Our Contributions}
\paragraph{Fully-Dynamic Dictionary for Random Multisets.}
We consider a new setting for dictionaries in which the dataset is a
random sample (with replacements) of the universe.  We refer to such
datasets as \emph{random multisets}.  We present the first
space-efficient fully-dynamic dictionary for random multisets that performs
inserts, deletes, and queries in  constant time
in the worst case whp.  The motivating application for random multi-sets
is in designing fully-dynamic filters (see Sec.~\ref{sec:reduce}). We note
that our dictionary can also maintain regular sets.

In the following theorem, we summarize the properties of our
dictionary, called the Crate Dictionary.  Note that we analyze the
number of memory accesses in the RAM model in which every memory
access reads/writes a word of $w = \Omega(\log n)$ contiguous bits.
All operations we perform in one memory access take constant time. 
We also analyze the probability that the space allocated for the
dictionary does not suffice; such an event is called an \emph{overflow}.

\begin{theorem}\label{thm:crate-dict}
The Crate Dictionary with parameter $n$ is a fully-dynamic dictionary that maintains sets and
random multisets with up to $n$ elements from a universe $\UU$ with
the following guarantees: (1)~For every polynomial in $n$ sequence
of insert, delete, and query operations, the dictionary does not
overflow whp. (2)~If the dictionary does not overflow, then every
operation (query, insert, delete) can be completed using at most
$4\cdot \frac{\log (|\UU|/n)}{w} + O(1)$ memory accesses. (3)~The
required space is $(1+o(1))\cdot n \log (|\UU|/n)+O(n)$ bits.
\end{theorem}

The comparable dictionary of
Arbitman~\etal~\cite{arbitman2010backyard} also achieves constant time
memory accesses, is space-efficient and does not overflow with high probability, but works
only for sets. We remark that dictionary constructions for
arbitrary multisets  
exist\cite{chazelle2004bloomier,DBLP:conf/soda/PaghPR05,demaine2006dictionariis,katajainen2010compact,patrascu2014dynamic} with weaker performance guarantees.

\paragraph{Fully-Dynamic Filter.} One advantage of supporting random multisets is that it allows us to
use our dictionary to construct a fully-dynamic filter. We show a reduction
that transforms a fully-dynamic dictionary on random multisets into a
fully-dynamic filter on sets (see Lemma~\ref{lemma:reduce dynamic}).
Applying this reduction to the Crate
Dictionary in Theorem~\ref{thm:crate-dict}, we obtain the Crate Filter in Theorem~\ref{thm:crate-filter}.
The Crate Filter is the first in-memory
space-efficient fully-dynamic filter that supports all operations in constant 
time in the worst case whp.
The filter does not overflow with high probability. We summarize its
properties in the following:\footnote{Note that the theorem holds for
all ranges of $\eps$, in particular, $\eps$ can depend on $n$ or
$w$. Moreover, even when $\eps=2^{-O(w)}$, each operation in the filter requires a constant number of memory accesses.}

\begin{theorem}\label{thm:crate-filter}
The Crate Filter with parameters $n$ and $\epsilon$ is a fully-dynamic filter that maintains a set of at most
$n$ elements from a universe $\UU$ with the following guarantees:
(1)~For every polynomial in $n$ sequence of insert, delete, and
query operations, the filter does not overflow whp. (2)~If the
filter does not overflow, then every operation (query, insert, and
delete) can be completed using at most
$4\cdot \frac{\log (1/\eps)}{w} + O(1)$ memory accesses. (3)~The
required space is $(1+o(1))\cdot n\log (1/\eps)+O(n)$ bits. (4)~For
every query, the probability of a false positive event  is
bounded by $\eps$.
\end{theorem}

\paragraph{Static Retrieval.}
We also present an application of the Crate Dictionary design for the
static Retrieval Problem~\cite{chazelle2004bloomier}.  In the static
retrieval problem, the input consists of a fixed dataset
$\DD\subseteq\UU$ and a function $f: \DD \rightarrow \set{0,1}^k$. On
query $x$, the output $y$ must satisfy $y=f(x)$ if $x\in \DD$ (if
$x\not\in \DD$, any output is allowed). The data structure can also
support updates in which $f(x)$ is modified for an $x\in \DD$.

We employ a variant of the Crate Dictionary design to obtain a Las
Vegas algorithm for constructing a practical data structure in linear
time whp.  The space requirements of this data structure are better
than previous constructions under certain parametrizations~(see
\cite{DBLP:conf/stacs/DietzfelbingerW19} and references therein).
In particular, the $o(1)$ term in the required space is similar to the one we obtain for our dictionary.

\begin{theorem} \label{thm:retrieval} There exists a data structure
for the static Retrieval Problem with the following guarantees:
(1)~Every query and update to $f(x)$ can be completed using at most
$4\cdot \frac{k}{w} + O(1)$ memory accesses. (2)~The time it takes
to construct the data structure is $O(n)$ whp. (3)~The required
space is $(1+o(1))\cdot nk + O(n)$ bits.
\end{theorem}

\subsection{Our Model}
\paragraph{Memory Access Model.} We assume that the data structures
are implemented in the RAM model in which the basic unit of one memory
access is a word. Let $w$ denote the memory word length in bits. We
assume that $w=\Omega(\log n)$. Performance is measured in terms of
memory accesses. We do not count the number of computer word-level
instructions because the computations we perform over words per memory
access can be implemented with $O(1)$ instructions using modern
instruction sets and lookup tables~\cite{Reinders2013, DBLP:conf/sigmod/PandeyBJP17,
bender2017bloom}.

We note that comparable dictionary designs assume words of
$\log(\size{\UU})$ bits~\cite{DBLP:conf/soda/PaghPR05, demaine2006dictionariis, arbitman2010backyard}. When the dictionaries are used to obtain filters~\cite{DBLP:conf/soda/PaghPR05, arbitman2010backyard},
they assume that the word length is $\log(1/\eps)$. If $w=\log(\size{\UU})$, then our designs always require a constant number of memory accesses per operation
(hence, constant time).  However, we prefer to describe our designs using a model in which the
word length does not depend on the size of the universe\cite{patrascu2014dynamic}.

\paragraph{Success Probability.}
Our constructions succeed with high probability in the following
sense. Upfront, space is allocated for the data structure and its
components. An \textit{overflow} is the event that the space allocated
for one of the components does not suffice (e.g., too many elements
are hashed to the same bin).  We prove that overflow occurs with
probability at most $1/\poly(n)$ and that one can control the degree
of the polynomial (the degree of the polynomial only
affects the $o(1)$ term in the size bound). In the case of random multisets, the
probability of an overflow is a joint probability distribution over
the random choices of the dictionary and the distribution over the
realizations of the multiset. In the case of sets, the probability of
an overflow depends only on the random choices that the dictionary/filter
makes.

\paragraph{Hash Functions.} We assume that we have access to perfectly
random hash functions and permutations with constant evaluation time.
We do not account the space these hash functions occupy.  A similar
assumption is made in~\cite{brodnik1999membership, raman2003succinct,
chazelle2004bloomier,dietzfelbinger2008succinct, bender2018bloom,
DBLP:conf/stacs/DietzfelbingerW19}. See~\cite{fredman1982storing,
kaplan2009derandomized,arbitman2010backyard,DBLP:conf/stoc/PatrascuT11, celis2013balls,meka2014fast}
for further discussion on efficient storage and evaluation of hash
functions.

\paragraph{Worst Case vs. Amortized.}
An interesting application that emphasizes the importance of
worst-case performance is that of handling search engine queries. Such
queries are sent in parallel to multiple servers, whose responses are
then accumulated to generate the final output. The latency of this
final output is determined by the slowest response, thus reducing the
average latency of the final response to the worst latency among the
servers. See~\cite{broder2001using,kirsch2007using, arbitman2009amortized,arbitman2010backyard}
for  further discussion on the shortcomings of expected or amortized
performance in practical scenarios.

\subsection{Main Ideas and Techniques}\label{sec:techniques}
We review the general techniques which we employ in our design and highlight the 
main issues that arise when trying to design a fully-dynamic filter. We then describe how our
construction addresses these issues.

\subsubsection{Relevant Techniques and Issues}

\paragraph{Dictionary Based Filters.} Carter
\etal~\cite{carter1978exact} observed that one can design a static
filter using a dictionary that stores hashes of the
elements. The reduction is based on a
random hash function $h:U\rightarrow \left[\frac{n}{\eps}\right]$. We
refer to $h(x)$ as the \emph{fingerprint} of $x$.  A static dictionary
over the set of fingerprints $h(\DD)$ becomes a static filter over the
set $\DD$.  Indeed, the probability of a false positive event is at
most $\eps$ because $\Pr{h(x)\in h(\DD)}\leq \eps$, for every
$x\notin\DD$. This reduction yields a space-efficient static filter if
the dictionary is space-efficient.

\paragraph{Delete Operations in Filters.}
Delete operations pose a challenge to the design of filters based on
dictionaries~\cite{DBLP:conf/soda/PaghPR05}.  Consider a collision,
namely $h(x)=h(y)$ for some $x \in \DD$ and $y\in\UU$. One must make
sure that a $\del(y)$ operation does not delete $x$ as well.  We
consider two scenarios.  In the first scenario, $y\notin \DD$ and a
$\del(y)$ causes $h(x)$ to be deleted from the filter. A subsequent
$\query(x)$ operation is then responded with a ``no'', rendering the
filter incorrect. To resolve this issue, we assume that a delete
operation is only issued for elements in the dataset.\footnote{This assumption is unavoidable if one wants to design a filter using the reduction
of Carter~\etal~\cite{carter1978exact}.}

\begin{assumption}\label{assume:del}
At the time a $\del(x)$ is issued to the filter, the
element $x$ is in the dataset.
\end{assumption}

In the second scenario, both $x\in\DD$ and $y\in\DD$ when $\del(y)$ is issued.
In this case, the
dictionary underlying the filter must be able to store duplicate
elements (or their multiplicities), so that $\del(y)$ deletes
only one duplicate leaving the other one intact (or subtracts the
counter).\footnote{This is why a fully-dynamic dictionary for
sets cannot be used out-of-the-box as a fully-dynamic filter.}  Pagh
\etal~\cite{DBLP:conf/soda/PaghPR05} propose a solution that requires
the underlying dictionary to support a multiset rather than a set.

\paragraph{Quotienting.} The space lower bound of $n\log(\size{\UU}/n)$ bits
suggests that one should be able to save $\log n$ bits in the representation
of each element. Indeed, the technique of
\textit{quotienting}~\cite{Knuth} saves $\log n$ bits by
storing them implicitly in an array of $n$ entries. This technique has
been successfully used in several filter and dictionary constructions
(see~\cite{pagh2001low,DBLP:conf/soda/PaghPR05, demaine2006dictionariis,
arbitman2010backyard, bender2012thrash,  bender2018bloom} and references
therein). The idea is to divide the $\log(\size{\UU})$ bits of an element $x\in \UU$ into
the first $\log n$ bits called the \emph{quotient} $q(x)$ (which are
stored implicitly),  and the remaining $\log(\size{\UU}/n)$ bits called
the \emph{remainder} $r(x)$ (which are stored explicitly).  An array
$A[0:n-1]$ is then used to store the elements, where $A[q]$ stores the
multiset of remainders $r(x)$ such that $q(x)=q$.\footnote{We abuse
notation and interpret $q(x)$ both as a number in $[n]$ and a binary
string of $\log n$ bits.} The benefit is that one does not need to
store the quotients explicitly because they are implied by the
location of the remainders.  

\paragraph{Load Balancing and Spares.} When quotients are
random locations in $[n]$, it is common to analyze dictionaries and
filters using the balls-into-bins
paradigm~\cite{wieder2017hashing}. In the case of dictionary/filter
design, the remainders are balls, and the quotients determine the
bins.  The balls that cannot be stored in the space allocated for
their bin are stored in a separate structure called the
\textit{spare}. A balls-into-bins argument implies that,
under certain parametrizations, the number of
balls stored in the spare is small compared to $n$ and can be
accommodated by the $o(1)$ term in the space
bounds~\cite{dietzfelbinger1990new, demaine2006dictionariis, pagh2008uniform,
arbitman2010backyard,bender2018bloom}.
Since the spare stores a sublinear number of elements, space
inefficient dictionaries can be employed to implement the spare.

Two challenges with this approach need to be addressed: (1)~Organize
the spare such that it can be accessed and updated in a constant
number of memory accesses. (2)~Manage the spare so that it does not
overflow. In particular, to avoid an overflow of the spare in the
fully-dynamic setting, one must make sure that balls are not only added to
the spare but also moved back to their bins.  For the balls-into-bins
analysis to remain valid over a sequence of insertions and deletions,
an invariant is maintained that requires that a ball be stored in the
spare only when its bin is full (Invariant~\ref{inv:SID}). This means
that whenever a ball from a full bin is deleted, one needs to move
another ball from the spare back to the bin.

The implementation of the spare plays a crucial role in whether the
dictionary supports deletions, duplicates, or operations in constant
time.  The dictionary of Arbitman \etal~\cite{arbitman2010backyard}
manages the spare as a de-amortized cuckoo hash
table~\cite{arbitman2009amortized} (so operations are in constant
time). Each time an element is relocated within the cuckoo hash table,
it checks if its corresponding bin is not full. If so, it leaves the
spare and returns to its bin (hence, deletions are
supported). However, the implementation of the spare
in~\cite{arbitman2010backyard} does not extend to filters with
deletions because their implementation does not support duplicate
elements.

\paragraph{Sparse vs. Dense Cases.}
Most dictionary and filter constructions consider two cases
(sparse/dense) based on the relative size of the dataset with respect
to the size of the universe (see~\cite{brodnik1999membership,
raman2003succinct,demaine2006dictionariis, arbitman2010backyard,
bender2018bloom}).  In the case of filters, the two cases are
separated based on the value of $\epsilon$. (See Def.~\ref{def:dense}).

In the dense case, elements are short enough (in bits) that a bin can pack
all the elements that hash into it within a word (so an element can be found 
by searching its corresponding bin in
constant time). The challenge in the sparse case is that the remainders
are long, and bins no longer fit in a constant number of words.
Arbitman~\etal~\cite{arbitman2010backyard} employ additional
structures in this case (i.e., global lookup tables) that point
to the location in which the remainder of an element is (separately) stored.

\subsubsection{Our Techniques}
\paragraph{Reduction Using Random Multisets.} We introduce a
relaxation of the Pagh \etal~\cite{DBLP:conf/soda/PaghPR05} condition
that a dictionary must support multisets to function as a fully-dynamic
filter (Sec.~\ref{sec:reduce}). This relaxation is based on the
observation that it suffices for the dictionary to support random multisets
for the reduction of Carter
\etal~\cite{carter1978exact} to succeed. Theorem~\ref{thm:crate-dict} provides
the first such dictionary. Applying the reduction to the Crate Dictionary gives 
Theorem~\ref{thm:crate-filter}. While random invertible permutations yield a reduction from a set to a random
set~\cite{demaine2006dictionariis, arbitman2010backyard} (see
Assumption~\ref{assume:rnd}), this is not the case for multisets.
Indeed, the image of a multiset with respect to a random permutation
is not a random multiset.

\paragraph{Pocket Dictionary.}

The starting point of our dictionary construction is the idea that the bins the
elements hash to (based on their random quotients) should behave like
self-contained dictionaries. Moreover, these dictionaries
should be space-efficient and fit in a constant number of words so that all operations on
them can be executed in a constant number of memory accesses.

Specifically, the basic building blocks of our design are small local
data structures that we employ in a modular, black-box fashion
(Sec.~\ref{sec:DictSmall}).  To this end, we construct a simple
space-efficient dictionary for small general multisets of elements whose
quotients belong to a limited range. We refer to this construction as
a \emph{pocket dictionary}. These small dictionaries function as bins;
elements choose a bin based on their (random) quotient.

Pocket dictionaries are space-efficient and use at most two extra bits
per element overall. This is an improvement over previous packing
techniques in which the number of bits per element is
$3$~\cite{Clerry,bender2012thrash} and
$17/8$~\cite{DBLP:conf/sigmod/PandeyBJP17}. We note that the practical
filter proposed in Pagh~\etal~\cite{DBLP:conf/soda/PaghPR05} also uses
two extra bits per element. However, they use a variant of linear
probing and therefore, the runtime of the operations is not constant (in the worst
case) and depends on the load. Arbitman
et. al~\cite{arbitman2010backyard} manage the bins by employing a global
lookup table for storing the encodings of all possible subsets stored
in a bin.

The simplicity of the pocket dictionary design gives it
flexibility. In particular, we also construct variants that support
variable-length remainders($\VarPD$) or count multiplicities of elements ($\CSD$).

\paragraph{Distributed Spares.}
To facilitate constant time access in the spare, we propose to manage
the spare by partitioning it across intervals of pocket dictionaries (hereafter
referred to as \emph{crates}). Specifically, we allocate one
distributed spare per crate such that each distributed spare stores
the overflow only from the pocket dictionaries in its crate (see Section~\ref{sec:SID}). With high probability,
at most $\poly (w)$ elements are stored in each distributed spare
(however, the elements no longer form a random
multiset). Each distributed spare is implemented as a \emph{space-inefficient
dictionary} ($\SID$).

The key advantage of distributed spares is that doubly linked lists of
elements from the same pocket dictionary can be implemented by storing pointers of length
$O(\log w)$ alongside the elements (as opposed to pointers of length
$\log n$ in the case of a global spare).  These linked lists enable
us to move an element from the spare back to its bin in a constant number
of memory accesses. 

The first technical challenge with this approach is that the elements
stored in a $\SID$ no longer form a random
multiset. To this end, we implement each $\SID$ as an
array of $\CSD$s that can maintain general multisets (See Appendix~\ref{sec:aux}).
The  second technical challenge  is that the probability of overflow of the components in each $\SID$  must be exponentially small in their size. (Indeed, their size is
$O(w)$ and the tolerated failure probability is at most
$1/\poly(n)$.) We formalize the properties of the $\SID$
in the following lemma. 

\begin{lemma}\label{lemma:crate-sid}
The $\SID$ with parameter $n'=\poly(w)$ is a fully-dynamic dictionary that maintains 
(general) multisets of cardinality at most $n'$ from a universe $\UU$ with
the following guarantees: (1)~For every polynomial in $n'$ sequence
of insert, delete, and query operations, the $\SID$ does not
overflow with probability at least $1-e^{-\Omega(w)}$. (2)~If the $\SID$ does not overflow, then every
operation (query, insert, delete) can be completed using at most
$2\cdot \frac{\log (|\UU|)}{w} + O(1)$ memory accesses. (3)~The
required space is $O(wn'\log(\size{\UU}))$ bits.
\end{lemma}

Note the extra $w$ factor in the space requirement that justifies the term space-inefficient dictionary. The $\SID$ also supports a special delete operation called pop, with the same performance guarantees (for an elaboration on the pop operation, see Sec.~\ref{sec:pop}).

\paragraph{Variable-Length Adaptive Remainders.}
In the sparse case, reading even one remainder of an element might require more
than $O(1)$ memory accesses. (The word length might be smaller than
$\log (\size{\UU}/n)$, the length of a remainder.) We propose to solve this problem by maintaining
variable-length prefixes of the remainders (Sec.~\ref{sec:adapt finger}). We call such prefixes
\emph{adaptive remainders}.  We maintain the adaptive remainders
dynamically such that they are minimal and prefix-free (Invariant~\ref{inv:min}).

The full remainder of an element is stored separately from its adaptive remainder but
their locations are synchronized. Specifically, the key property of the adaptive remainders is that they allow us to find the location of an element $x$ using a constant number of memory
accesses, regardless of the length of $x$.\footnote{The ``location''
of an element $x$ is ambiguous since $x$ might not even be in the
dataset. More precisely, we can certify whether $x$ is in the
dataset or not by reading at most two full remainders and these full
remainders can be located in a constant number of memory accesses
(see Claim~\ref{claim:ptr}).} Indeed, one may view this technique
as a method for dynamically maintaining a perfect hashing of the
elements in the same bin with the same
quotient. As opposed to typical perfect hashing schemes in which the images of
the elements are of the same length, we employ variable-length
images. The adaptive remainders have, on average, constant length so
they do not affect space efficiency. Indeed, considering variable-length
adaptive remainders allows us to bypass existing space bounds on dynamic perfect hash functions
that depend on the size of the universe~\cite{mortensen2005dynamic}. 

Prefix-free variable-length adaptive fingerprints are employed in the adaptive Broom filter of Bender
\etal~\cite{bender2018bloom} to fix false positives and maintain distinct
fingerprints for elements in the filter. Their adaptive fingerprints are computed
by accessing an external memory reverse hash table that is not
counted in the filter's space.~\footnote{The reverse hash table allows the filter to look up an element based on its fingerprint in the filter.~\cite{bender2018bloom}.} We emphasize that all our adaptive
remainders are computed in-memory.

\subsection{Related Work}
The topic of dictionary and filter design is a fundamental
theme in the theory and practice of data structures.
 We restrict our focus to the results that
are closest to our setting.
\paragraph{Dictionaries.}
To the best of our knowledge, the dictionary of Arbitman~\etal~\cite{arbitman2010backyard} is the
only space-efficient fully-dynamic dictionary for sets that performs
queries, insertions, and deletions in constant time in the worst case with high probability.

Several other fully-dynamic constructions
support operations in constant time with high
probability~\cite{ dietzfelbinger1990new, dalal2005two,
demaine2006dictionariis,
arbitman2009amortized} but are not
space-efficient. On the other hand, some dictionaries are
space-efficient but do not have constant time guarantees with high
probability for all of their 
operations~\cite{raman2003succinct, fotakis2005space,panigrahy2005efficient,dietzfelbinger2007balanced}. For the static case, several space-efficient constructions exist
that perform queries in constant time~\cite{brodnik1999membership,pagh2001low,patrascu2008succincter}.

\paragraph{Filters.} In the context of filters, to the best of our
knowledge, only a few space-efficient constructions have proven
guarantees. We focus here on space-efficient filters that perform
operations in a constant number of memory accesses (albeit with
different probabilistic guarantees). The fully-dynamic filter of Pagh
\etal~\cite{DBLP:conf/soda/PaghPR05} supports constant time queries
and insertions and deletions in amortized expected constant time. The
incremental filter of Arbitman~\etal~\cite{arbitman2010backyard}
performs queries in constant time and insertions in constant time with
high probability. The construction of Bender et
al.~\cite{bender2018bloom} describes an adaptive filter~\footnote{Loosely
speaking, an adaptive filter is one that
fixes false positives after they occur~\cite{bender2018bloom, mitzenmacher2018adaptive}.} equipped with
an external memory reverse hash table that supports queries,
insertions and deletions in constant time with high probability. (The
space of the reverse hash table is not counted in the space of their
filter.) Space-efficient filters for the static case have been studied
extensively~\cite{mitzenmacher2002compressed,dietzfelbinger2007balanced,dietzfelbinger2008succinct, porat2009optimal}. Several
heuristics such as the cuckoo filter~\cite{DBLP:conf/conext/FanAKM14},
the quotient filter~\cite{bender2012thrash,
DBLP:conf/sigmod/PandeyBJP17} and variations of the Bloom
filter~\cite{bloom1970spacetime} have been reported to work well in
practice.

\subsection{Paper Organization}
The preliminaries are in Section~\ref{sec:prelim}.  In
Section~\ref{sec:reduce}, we discuss the reduction from a fully-dynamic
dictionary on random multisets to a fully-dynamic filter on sets. In
Section~\ref{sec:DictSmall}, we describe the structure of the pocket
dictionary and its variants. We also briefly mention some auxiliary
data structures which we employ and whose complete description can be
found in Appendix~\ref{sec:aux}.

The description and analysis of the Crate Dictionary are covered in two
parts. In Section~\ref{sec:cratefilter}, we discuss the Crate
Dictionary construction and analysis in the dense case. The
distributed spares are described in
Section~\ref{sec:SID}. Section~\ref{sec:crate-dict-sparse} covers the
Crate Dictionary in the sparse case.
Section~\ref{sec:retrieval}
describes the modifications required to perform static retrieval.

\section{Preliminaries}\label{sec:prelim}

\paragraph{Notation.}
The \emph{indicator function} of a set $S$ is the function
$\mathds{1}_S: S\rightarrow \set{0,1}$ defined by
\begin{align*}
\mathds{1}_S(x)&\triangleq
\begin{cases}
1 & \text{if $x\in S$},\\
0 & \text{if $x\not\in S$}\;.
\end{cases}
\end{align*}
For any positive $k$, let $[k]$ denote the set
$\set{0,\ldots,\ceil{k}-1}$.  For a string $a \in \set{0,1}^*$, let $\size{a}$ denote the length
of $a$ in bits.

We define the range of a hash function $h$ to be a set of natural numbers
$[k]$ and also treat the image $h(x)$ as a
binary string, i.e., the binary representation of $h(x)$ using
$\ceil{\log_2 k}$ bits.

Given two strings, $a,b\in \set{0,1}^*$, the concatenation of $a$ and
$b$ is denoted by $a\circ b$. We denote the fact that $a$ is a
prefix of $b$ by $a\subseteq b$.

\begin{definition}
A sequence of strings $A \triangleq \set{a_1,\ldots, a_s}$ is
\emph{prefix-free} if, for every $i\neq j$, the string $a_i$ is not a
prefix of $a_j$.
\end{definition}

\subsection{Filter and Dictionary Definitions}
Let $\UU$ denote the universe of all possible elements.
\paragraph{Operations.}

We consider three types of operations:
\begin{itemize}
\item $\ins(x_t)$ - insert $x_t\in \UU$ to the dataset.
\item $\del(x_t)$ - delete $x_t\in \UU$ from the dataset.
\item $\query(x_t)$ - is $x_t\in \UU$ in the dataset?
\end{itemize}

\paragraph{Dynamic Sets and Random Multisets.}
Every sequence of operations $R=\set{\op_{t}}_{t=1}^{T}$ defines a
\emph{dynamic set} $\DD(t)$ over $\UU$ as follows.\footnote{ The
definition of state in Equation~\ref{eq:state} does not rule out a
deletion of $x\notin\DD(t-1)$.  However, we assume that
$op_t=\del(x_t)$ only if $x_t\in \DD(t-1)$. See
Assumption~\ref{assume:del} and the discussion therein. 
}

\begin{align}\label{eq:state}
\DD(t)&\triangleq
\begin{cases}
\emptyset & \text{if $t=0$}\\
\DD(t-1)\cup \set{x_t} &\text{if $\op_t=\ins(x_t)$}\\
\DD(t-1)\setminus \set{x_t} &\text{if $\op_t=\del(x_t)$}\\
\DD(t-1)& \text{if $t>0$ and $\op_t=\query(x_t)$.}
\end{cases}
\end{align}

\begin{definition}
A \emph{multiset} $\MM$ over $\UU$ is a function $\MM:\UU\rightarrow
\NN$. We refer to $\MM(x)$ as the \emph{multiplicity} of $x$. If
$\MM(x)=0$, we say that $x$ is not in the multiset. We refer to
$\sum_{x\in \UU} \MM(x)$ as the \emph{cardinality} of the multiset
and denote it by $\size{\MM}$.
\end{definition}

The \emph{support} of the multiset is the set
$\set{x \mid \MM(x)\neq 0}$. The \emph{maximum multiplicity} of a multiset is
$\max_{x\in\UU} \MM(x)$.

A \emph{dynamic multiset} $\set{\MM_t}_t$ is specified by a sequence
of insert and delete operations. Let $\MM_t$ denote the multiset after
$t$ operations.\footnote{As in the case of fully-dynamic sets, we require
that $\op_t=\del(x_t)$ only if $\MM_{t-1}(x_t)>0$.}
\begin{align*}
\MM_t(x)&\triangleq
\begin{cases}
0 & \text{if $t=0$}\\
\MM_{t-1}(x)+1 &\text{if $\op_t=\ins(x)$}\\
\MM_{t-1}(x)-1 &\text{if $\op_t=\del(x)$}\\
\MM_{t-1}(x)& \text{otherwise.}
\end{cases}
\end{align*}
\medskip\noindent We say that a dynamic multiset $\set{\MM_t}_t$ has
cardinality at most $n$ if $\size{\MM_t}\leq n$, for every $t$.
\begin{definition}
A dynamic multiset $\MM$ over $\UU$ is a\emph{ random
multiset} if for every $t$, the multiset $\MM_t$ is the outcome of 
independent uniform samples (with replacements) from $\UU$.
\end{definition}

\begin{comment}
An oblivious adversary that generates a fully-dynamic multi-set of
cardinality at most $n$ is specified by a sequence of operations
$\set{\op_t}_t$, where each operation is of the form:
$\ins$, $\del(t')$, $\query(x)$
that satisfy the
following conditions:
\begin{enumerate}
\item If $\op_t=\ins$, then a random element $x_t\in\UU$ is
inserted to the dataset.
\item If $\op_t$
\end{enumerate}

The following observation states that a fully-dynamic multiset generated by
an oblivious adversary remains random at all times (even after
deletions).

as long as the insertions are independent samples.
\begin{observation}
Consider a sequence of operations in which the sequence 
$\{x_t | \op_t=\ins(x_t)\}$ is a sequence of independent elements
drawn uniformly at random from $\UU$.  Then $\MM(t')$ is a random
multiset for every $t'$.
\end{observation}
\end{comment}
\paragraph{Fully-Dynamic Filters.}
A \emph{fully-dynamic filter} is a data structure that maintains a dynamic set
$\DD(t)\subseteq \UU$ and is parameterized by an error parameter
$\eps\in (0,1)$.  Consider an input sequence that specifies a dynamic
set $\DD(t)$, for every $t$.  The filter outputs a bit for every query
operation. We denote the output that corresponds to $\query(x_t)$ by
$\out_t\in\set{0,1}$.  We require that the output satisfy the
following condition:
\begin{align}
\op_t&=\query(x_t) \Rightarrow \out_t \geq \indD(x_t)\;.
\end{align}
The output $\out_t$ is an approximation of $\indD(x_t)$ with a one-sided
error.  Namely, if $x_t\in \DD(t)$, then $b_t$ must equal $1$.

\begin{definition}[false positive event]
Let $FP_t$ denote the event that $\op_t=\query(x_t)$,
$\out_t=1$ and $x_t\notin\DD(t)$.
\end{definition}

\medskip\noindent
The error parameter $\eps\in (0,1)$ is used to bound the probability
of a false positive error.
\begin{definition}
We say that the \emph{false positive probability} in a filter is
bounded by $\eps$ if it satisfies the following property. For every
sequence $R$ of operations and every $t$,

\begin{align*}
\Pr{\FP_t}
\leq\eps\;.
\end{align*}
\end{definition}
The probability space in a filter is induced only by the random
choices (i.e., choice of hash functions) that the filter makes. Note
also that if $\op_t=\op_{t'}=\query(x)$, where
$x\not\in \DD(t)\cup \DD(t')$, then the events $\FP_{t}$ and
$\FP_{t'}$ may not be independent
(see~\cite{bender2018bloom,mitzenmacher2018adaptive} for a discussion
of repeated false positive events and adaptivity).

\paragraph{Fully-Dynamic Dictionaries.}
A \emph{fully-dynamic dictionary} with parameter $n$ is a
fully-dynamic filter with parameters $n$ and $\eps=0$. In the case of multisets, the response $\out_t$ of a fully-dynamic
dictionary to a $\query(x_t)$ operation must satisfy  $\out_t=1$ iff $\MM_t(x_t)>0$.
\footnote{ One may also define $\out_t=\MM_t(x_t)$.}

When we say that a filter or a dictionary has parameter $n$, we mean that the cardinality of the input set/multiset is at most $n$ at all points in time.

\paragraph{Success Probability and Probability Space.} We say that a
dictionary (filter) \emph{works} for sets and random multisets if the
probability that the dictionary does not overflow is high (i.e., it is
$\geq 1-1/\poly(n)$). The probability in the case of random multisets
is taken over both the random choices of the dictionary and the
distribution of the random multisets. In the case of sets, the success
probability depends only on the random choices of the dictionary.

\paragraph{Relative Sparseness.}
The relative sparseness~\cite{brodnik1999membership} of a set (or
multiset) $\DD\subseteq \UU$ is the ratio $\size{\UU}/\size{\DD}$.  We
denote the relative sparseness by $\rho$. In the fully-dynamic setting, we
define the relative sparseness to be $\rho\triangleq \size{\UU}/{n}$,
where $n$ is the upper bound on $\size{\DD(t)}$.

Recall that $w=\Omega(\log n)$ denotes the memory word length.  We
differentiate between two cases in the design of dictionaries, depending on
the ratio $w/\rho$. 
\begin{definition}\label{def:dense}
The \emph{dense case} occurs when $w/\rho=\Omega_n(1)$.  The
\emph{sparse case} occurs when $w/\rho=o_n(1)$.\footnote{By
$w/\rho=o_n(1)$, we mean that
$\lim_{n\to \infty} \frac{w}{\rho} = 0$. By $w/\rho=\Omega_n(1)$,
we mean that there exists a constant $c>0$ such that
$w/\rho \geq c$, for sufficiently large values of $n$. }
\end{definition}

The reduction from dictionaries (see Sec~\ref{sec:reduce}) implies
that the dictionaries used to implement the filter have a relative
sparseness $\rho = 1/\eps$. Hence, in the case of filters, we
differentiate between the \emph{high error} case, in which to
$\eps=\Omega_n(1/w)$, and the \emph{low error} case in which
$\eps=o_n(1/w)$.

\section{Reduction: Filters Based on Dictionaries}\label{sec:reduce}
In this section, we extend the reduction of
Carter~\etal~\cite{carter1978exact} to construct fully-dynamic filters out
of fully-dynamic dictionaries for random multisets. Our reduction can be
seen as a relaxation of the reduction of
Pagh~\etal~\cite{DBLP:conf/soda/PaghPR05}. Instead of requiring that
the underlying dictionary support multisets, we require that it only
supports random multisets.

\begin{claim}
Consider a random hash function
$h:\UU \rightarrow \left[\frac{n}{\eps}\right]$ and let
$\DD(t)\subseteq \UU$ denote a dynamic set whose cardinality never
exceeds $n$. Then $h(\DD(t))$ is a random multiset of cardinality at
most $n$.
\end{claim}
Since $h$ is random, an ``adversary'' that generates the sequence of
insertions and deletions for $h(\DD(t))$ is an oblivious adversary in
the following sense. When inserting, it inserts a random element
(which may be a duplicate of a previously inserted
element\footnote{Duplicates in $h(\DD(t))$ are caused by collisions
(i.e., $h(x)=h(y)$) rather than by reinsertions.}). When deleting at
time $t$, it specifies a previous time $t'<t$ in which an insertion
took place, and requests to delete the element that was inserted at
time $t'<t$.

Let $\Dict$ denote a fully-dynamic dictionary for random multisets of
cardinality at most $n$ from the universe
$\left[\frac{n}{\eps}\right]$. The following lemma proves that Theorem~\ref{thm:crate-filter} 
is implied by Theorem~\ref{thm:crate-dict}.
\begin{lemma}\label{lemma:reduce dynamic}
For every dynamic set $\DD(t)$ of cardinality at most $n$, the
dictionary $\Dict$ with respect to the random multiset $h(\DD(t))$
and universe $\left[\frac{n}{\eps}\right]$ is a fully-dynamic filter for
$\DD(t)$ with parameters $n$ and $\eps$.
\end{lemma}
\begin{proofsketch}
The $\Dict$ records the multiplicity
of $h(x_t)$ in the multiset $h(\DD(t))$ and so deletions are performed correctly. The filter outputs $1$ if
and only if the multiplicity of $h(x_t)$ is positive. False positive events are caused by
collisions in $h$. Therefore, the
probability of a false positive is bounded by $\eps$ because of the
cardinality of the range of $h$.
\end{proofsketch}

\section{The Pocket Dictionary}\label{sec:DictSmall}

In this section, we propose a fully-dynamic dictionary construction for the
special case of small multisets of elements consisting of
quotient/remainder pairs $(q_i,r_i)$. This construction works provided
that the quotients belong to a small interval. We emphasize that the
data structure works for multisets in general, without the assumption
that they are random.

Let $F\triangleq \set{(q_i,r_i)}_{i=0}^{f-1}$ denote the multiset of
quotient/remainder pairs to be stored in the dictionary, with
$q_i\in [m]$ and $|r_i|=\ell$, for every $i\in[f]$. Namely, $m$ is the
length of the interval that contains the quotients $q_i\in [m]$, $f$
is an upper bound on the number of elements $(q_i,r_i)$ that will be
stored in the dictionary, and $\ell$ is the length in bits of each
remainder $r_i$.  We store $F$ in a data structure called a
\emph{pocket dictionary} denoted by $\PD(m,f,\ell)$.

The pocket dictionary data structure uses two binary strings, denoted
by $\hheader(F)$ and $\bbody(F)$, as follows.  Let
$n_q\triangleq |\set{i \in [f] \mid q_i=q}|$ denote the number of
elements that share the same quotient $q$.  The header $\hheader(F)$
stores the vector $(n_0,\ldots,n_{m-1})$ in unary as the string
$1^{n_0}\circ 0 \circ \cdots \circ 1^{n_{m-1}}\circ 0$. The length of
the header is $m+f$.  The body $\bbody(F)$ is the concatenation of
remainders $r_i$ sorted in nondecreasing lexicographic order of
$\set{(q_i,r_i)}_{i\in[f]}$. The length of the body is
$f\cdot \ell$. We refer the reader to Figure~\ref{fig:pocket} for a
depiction.

\begin{figure}[h]

\centering
\includegraphics[width=\textwidth]{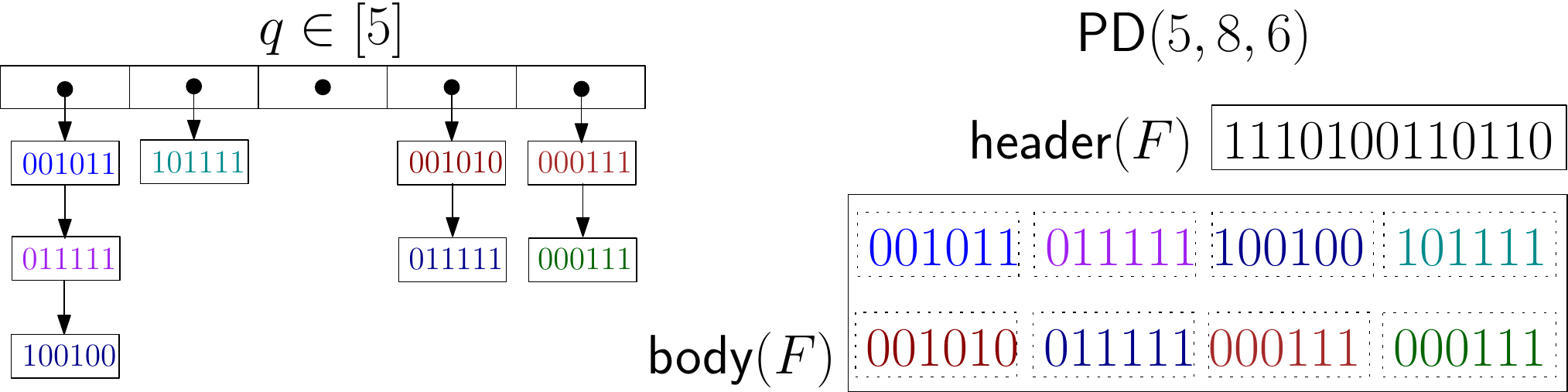}
\caption{\textbf{Pocket Dictionary Encoding.} Given $F = \{(0,001011),(0,011111),(0,100100),\\(1,101111), (3,001010),(3,011111),(4,000111),(4,000111)\}$, on the left we have its open address hash table representation and on the right, the way it is encoded in a $\PD(5,8,6)$.}
\label{fig:pocket}
\end{figure}

\medskip\noindent
Let $|\PD(m,f)|= |\hheader(F)|+|\bbody(F)|$ denote the number
of bits used to store $F$.
Recall that the word length $w$ satisfies $w=\Omega (\log n)$.

\begin{claim}\label{claim:DictSmall size} The number of bits that
$\PD(m,f,\ell)$ requires is
$m+f(1+\ell)$.  If $f=O(m)$, $\ell\leq w$ and
$m=O\parentheses{\frac{w}{\ell}}$, then
$\PD(m,f,\ell)$
fits in a constant number of words.
\end{claim}

One can decode $F$ given $\hheader(F)$ and $\bbody(F)$. Moreover,
one can modify  $\hheader(F)$ and $\bbody(F)$ to reflect
updates (insertions and deletions) to $F$. When the pocket dictionary
fits in a constant number of words, all operations on the multiset
can, therefore, be executed using a constant number of memory
accesses.

We say that a pocket dictionary $\PD(m,f,\ell)$ overflows
if more than $f$ quotient/remainder pairs are stored in it.

\subsection{Variable-Length Pocket Dictionary}\label{sec:var pocket dict}
The use of variable-length adaptive remainders for the sparse case requires a variant of the
pocket dictionary that supports variable-length values. We refer to this construction as the \emph{variable-length pocket dictionary} ($\VarPD$). This variant
shares most of its design with the pocket dictionary with fixed size remainders. We briefly describe the required modifications.

We denote this modified pocket dictionary by
$\VarPD(M,F,L)$. This new dictionary is in charge of storing at
most $F$ variable-length remainders $\set{r_i}_{i \in [F]}$ with
quotients $q_i\in[M]$, where $\sum_{i\in [F]} \size{r_i}\leq L$.  The
header of the dictionary is the same as in $\PD(M, F, L)$. The
body consists of the list of remainders separated by an
``end-of-string'' symbol. The list of remainders appears in the body
in nondecreasing lexicographic order of the strings $(q_i, r_i)$. We
suggest to simply use two bits to represent the ternary alphabet, so 
the space requirement of
the pocket dictionary for variable-length remainders at most doubles.

\begin{claim}\label{claim:varpd}
The number of bits that $\VarPD(M,F,L)$ requires is
$\size{\VarPD(M,F,L)} \leq M+3F+2L$. If
$\max\set{M,F,L} = O\parentheses{w}$, then $\VarPD(M,F,L)$ fits in
a constant number of words.
\end{claim}

We say that a variable-length pocket dictionary
$\VarPD(M, F, L)$ overflows if the cardinality of the multiset
exceeds $F$ or the total length of the remainders exceeds $L$.

\subsection{Auxiliary  Data Structures}

We employ two additional data structures for storing elements whose
pocket dictionaries overflow. The first data structure is a
space-inefficient variant of the pocket dictionary, called a
\emph{counting set dictionary} ($\CSD$). Its variable-length
counterpart is called a \emph{variable-length counting set dictionary}
($\VarCSD$). These data structures need not be space-efficient since,
across all distributed spares, they store only a sublinear number of elements. A key property of the
$\CSD$ is that its failure probability (i.e. overflow) is
exponentially small in its capacity. The implementation details of
these auxiliary data structures appear in Appendix~\ref{sec:aux}.

\section{Crate Dictionary - Dense Case}\label{sec:cratefilter}

In this section, we present a fully-dynamic dictionary for sets and random
multisets for the case in which the relative sparseness satisfies
$w/\rho=\Omega_n(1)$. We refer to this construction as the Crate
Dictionary for the dense case. 

\subsection{Structure}

\begin{table}[t]
\centering
\begin{tabular}{|c|c|l|}
\hline
symbol& value & role \\
\hline\hline 
$\mu$& $\mumacro$ & extra capacity (over the average) per $\PD$
\rule{0pt}{3.6ex}\\
$\ell$& $\log \rho$&  length of $r(x)$ \\
$m$&  $\frac{w}{\ell}$& interval of quotients per $\PD$\\
$f$& $(1+\mu) m$& capacity of $\PD$\\
$\beta$& $\betaval$& exponent affecting number of crates, for $\delta>0$\\
$\beta'$& $\betapval$ & exponent affecting size of $\SID$
\\
$\CC$ & $w^{\beta}$ & $n/\CC$ is the number of crates.\\
$\CC'$ & $w^{\beta'}$ & capacity of the $\SID$\\
\hline
\end{tabular}
\caption{Internal parameters of a Crate Dictionary with parameter
$n$. Recall that $\rho = \size{U}/n.$}
\label{tbl:crate}
\end{table}

The Crate Dictionary is parametrized by the cardinality $n$
of the dynamic random multiset. The dictionary consists of
two levels of dictionaries.  The dictionaries in the first level are
called \emph{crates}.  Each crate consists of two main parts: a set of
pocket dictionaries (in which most of the elements are stored) and a
space-inefficient dictionary ($\SID$) that stores elements whose
corresponding pocket dictionary is full. The internal parameters of
the Crate Dictionary are summarized in Table~\ref{tbl:crate}.  See
Fig.~\ref{fig:crate} for a depiction of the
structure of each crate.

The number of crates is $n/\CC$. 
The number of pocket dictionaries in each crate is
$\CC/{m}$, and each pocket dictionary is a
$\PD(m,f,l)$. The space inefficient dictionary $\SID(\tf,\tl)$
stores at most $\tf$ elements, each of length $\tl$ bits.  Each $\SID$ is
implemented using counting set dictionaries (Sec.~\ref{sec:count set
dict}). We elaborate on the structure and functionality of the
$\SID$ separately in Sec.~\ref{sec:SID}.

\begin{figure}[h]
\centering
\includegraphics[width=\textwidth]{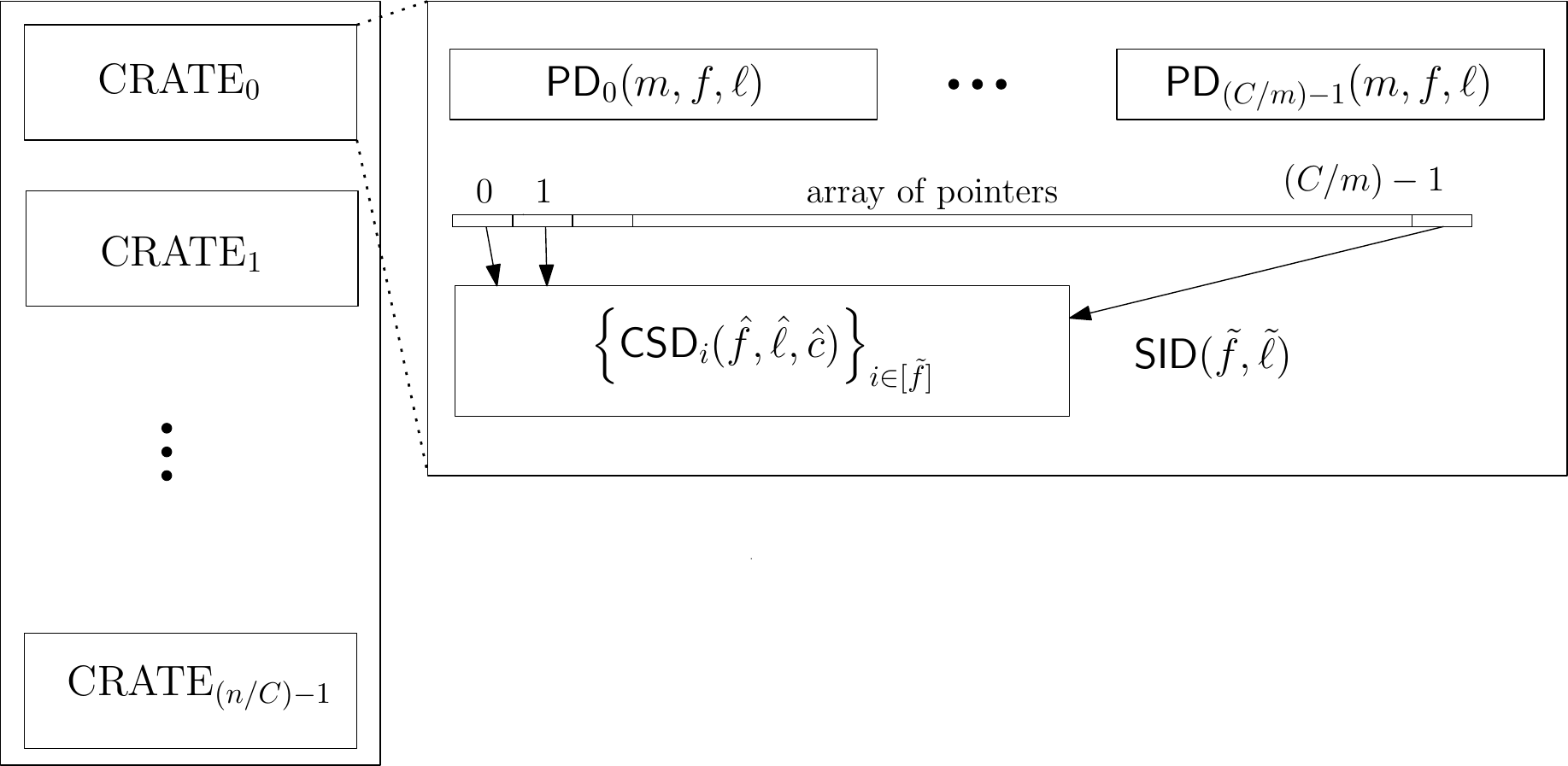}
\caption{\textbf{Crate Structure.}  Each crate has $\CC/m$
$\PD$'s, an array of $\CC/m$ pointers (i.e, one
pointer per $\PD$) and a $\SID(\tf, \tl)$. The
$\SID$ consists of $\tf$ $\CSD$'s that  store
elements and pointers.}
\label{fig:crate}
\end{figure}
\subsection{Element Representation}

Without loss of generality, we have that $\UU=[\rho n]$.  Since the multiset is
random, its elements are chosen with replacements u.a.r. and independently from
$[\rho n]$.  We represent each $x\in [\rho n]$ by a $4$-tuple
$(\hc(x),\hb(x),q(x),r(x))$, where
\begin{align*}
h^{c}(x)&\triangleq\floor*{\frac{x}{\ceil*{\rho\cdot \CC}}}\in \left[\frac{n}{\CC}\right] &\text{(crate index)}\\
h^{b}(x)&\triangleq \floor*{\frac{x(\bmod\ceil*{\rho\cdot \CC})}{\ceil*{\rho m}}}
\in \left[\frac{\CC}{m}\right] &\text{($\PD$ index)}\\
q(x)&\triangleq \floor*{\frac{
x(\bmod\ceil*{\rho\cdot \CC}) \bmod(\ceil{\rho m})}
{\ceil{\rho}}} \in [m] & \text{($\PD$ quotient)}\\
r(x)&\triangleq
x(\bmod\ceil*{\rho\cdot \CC}) \bmod(\ceil{\rho m}) \bmod(\ceil{\rho})
\in\left[\rho\right]&\text{($\PD$ remainder)}
\end{align*}

Informally, consider the binary representation of $x$, then $r(x)$ is
the least significant $\log \rho$ bits, $q(x)$ is the next $\log m$
bits, and so on. 

An element $x$ is stored in the crate of index $\hc(x)$.  Within its
crate, the element $x$ is stored in one of two places: (1)~in the pocket
dictionary of index $\hb(x)$ that stores $(q(x),r(x))$, or  (2)~in the
$\SID$ that stores $(\hb(x),q(x),r(x))$.  (Storing $\hb(x)$ in the
$\SID$ avoids spurious collisions since elements from all the full
pocket dictionaries in the crate are stored in the same $\SID$.)

\subsection{Functionality}
An operation on element $x\in\UU$ is forwarded to the crate whose
index is $\hc(x)$. We elaborate here on how a crate $\hc(x)$ supports
queries, insertions, and deletions involving an element
$x$. Operations to the $\SID$ are described in Sec.~\ref{sec:SID}.

A $\query(x)$ is implemented by searching for $(q(x),r(x))$ in the
pocket dictionary of index $\hb(x)$ in the crate. If the pocket
dictionary is full and the element has not been found, the query
is forwarded to the $\SID$.

An $\ins(x)$ operation first attempts to insert $(q(x),r(x))$ in the
pocket dictionary of index $\hb(x)$ in the crate. If the pocket
dictionary is full, it forwards the insertion to the $\SID$.

To make sure that the $\SID$ does not overflow whp in the fully-dynamic setting,
we maintain the following invariant.
\begin{invariant}\label{inv:SID}
An element $y$ is stored in the $\SID$ only if
the pocket dictionary of index $\hb(y)$ in crate $\hc(y)$ is full.
\end{invariant}

As a consequence, an element stored in the $\SID$ must be moved to its
corresponding pocket dictionary whenever this pocket dictionary is no
longer full due to a deletion. We refer to the operation of retrieving
such an element from the $\SID$ as a \emph{pop} operation
(Sec.~\ref{sec:pop}).

To maintain Invariant~\ref{inv:SID}, a $\del(x)$ differentiates
between the following cases depending on where  $x$
is stored and whether its pocked dictionary is full:
\begin{enumerate}
\item If $(q(x),r(x))$ is in the pocket dictionary of index $\hb(x)$ and
the pocket dictionary was not full before the deletion, we simply
delete $x$ from the pocket dictionary and return.
\item If $(q(x),r(x))$ is in the pocket dictionary of index $\hb(x)$
and the pocket dictionary was full before the deletion, we delete
$(q(x),r(x))$ from the pocket dictionary and issue a $\pop(\hb(x))$
to the $\SID$. The pop operation returns a triple
$(\hb(y), q(y), r(y))$, where $\hb(y) = \hb(x)$ (if any). We then
insert $(q(y), r(y))$ into the pocket dictionary of index $\hb(y)$
and return.
\item If $(\hb(x),q(x),r(x))$  is found in the $\SID$ of crate $q_c(x)$, we
delete it from the $\SID$.
\end{enumerate}
Note that duplicate elements must be stored to support random
multisets. In particular, a deletion should erase only one duplicate
(or decrease the counter in the $\SID$ appropriately). This
description allows for copies of the same element to reside both in
the $\SID$ and outside the $\SID$. Precedence is given to storing
elements in the pocket dictionaries because the $\SID$ is used only for
elements that do not fit in their pocket dictionaries due to
overflow. (One could consider a version in which deletions and
insertions of duplicates are first applied to the $\SID$.)

\subsection{A Space-Inefficient Dictionary ($\SID$)}\label{sec:SID}
The space-inefficient dictionary ($\SID$) is a fully-dynamic dictionary for
arbitrary dynamic multisets. It is used for storing elements whose pocket
dictionaries are full. A $\SID$ supports queries, insertions,
deletions, and pop operations (see Sec.~\ref{sec:pop} for a
specification of pop operations).  A key property of a $\SID$ is that
the probability of overflow is at most $1/\poly(n)$. Note that this
bound on the probability of overflow can be exponentially small in the
cardinality of the multiset stored by a $\CSD$ in the $\SID$.%
\footnote{Maybe previous dictionary designs such
as~\cite{patrascu2014dynamic, demaine2006dictionariis} can be parameterized to work with
exponentially small failure probabilities.}%
\footnote{The probability space over which overflow is bounded depends
only on the random hash function of the $\SID$.}
The explicit bits of an element $x$ stored in the $\SID$ are
$(\hb(x),q(x),r(x))$.  See Fig.~\ref{fig:SID} for a depiction.

\paragraph{Counting Set Dictionaries.}
The basic building block of the $\SID$ is a data structure called the \emph{counting set dictionary} ($\CSD$).
Each $\CSD$ stores a multiset of elements. The data structure is parametrized by the cardinality of the
support set of the multiset and by the maximum multiplicity of elements in the multiset.
The $\CSD$ stores $(\hb(x),q(x),r(x))$ and pointers of length $O(\log w)$ used to support $\pop$ operations.
More details on the $\CSD$ can be found in Appendix~\ref{sec:count set
dict}.

\paragraph{Structure.} The parameters used for the design of the
$\SID$ are summarized in Table~\ref{tbl:sid parameters}.
\begin{table}
\centering
\begin{tabular}{|c|c|l|}
\hline
symbol& value & role \\
\hline\hline 
$\tf$& $O(C')$& bound on the cardinality of the multiset and\rule{0pt}{2.6ex}\\
	&& the number of $\CSD$s per $\SID$\\
$\tl$& $O(\log w)$& length of $(\hb(x),q(x),r(x))$\\
$\hf$&  $cw/\hl$& cardinality of the support of the multiset in $\CSD$\\
$\hl$& $\tl+O(\log w)$ & length element plus pointers\\
$\hatc$& $\tf$ & bound on multiplicity of element in $\CSD$\\
c& $O(1)$ & constant affecting exponent in overflow probability
\\
\hline
\end{tabular}
\caption{Parametrization of the $\SID$.}
\label{tbl:sid parameters}
\end{table}
A $\SID$ consists of $\tf$ counting set dictionaries
$\CSD(\hf,\hl,\hatc)$.\footnote{We pessimistically set the upper bound $\hatc$
on the  maximum multiplicity of any element in a $\CSD$ to be the
cardinality $\tf$ of the dynamic multiset stored in the $\SID$.}  By
Claim~\ref{claim:CSD}, each $\CSD$ fits in a constant number of words.
An independent fresh random hash function
$\thash:{\UU} \rightarrow [\tf]$ is utilized to map each element to
one of the $\CSD$s in the $\SID$. We note that
$\thash$ must be independent of $\hb$ for the analysis of the
maximum loads in each $\CSD$ to work (the same hash
function $\thash$ may be used for all $\SID$s).\footnote{We note that if
$(\hb(x), q(x),r(x)) = (\hb(y), q(y),r(y)) $ are stored in the same
$\SID$, then $x=y$.}

\begin{figure}[h]
\centering
\includegraphics[width=\textwidth]{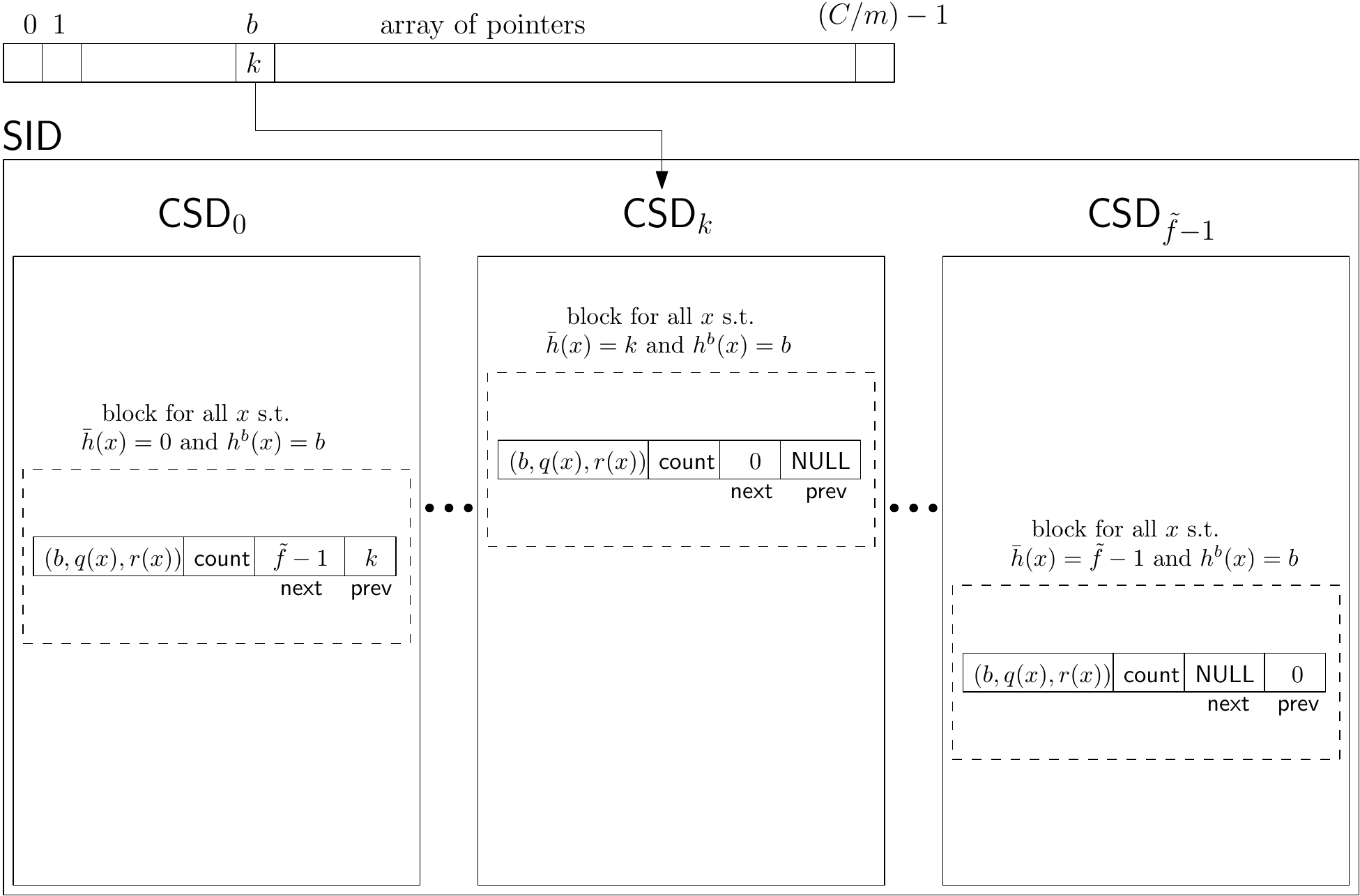}
\caption{\textbf{$\SID$ Structure.} Each $\CSD_i$ explicitly stores $(\hb(x),q(x),r(x))$
contiguously    in non-decreasing lexicographic order for
every element $x$ with $\thash(x)=i$ that is not stored in the $\PD$s. We remark that no attempt
to save space is made since the $\SID$ stores a sublinear number of elements. 
Each distinct element stored in the $\CSD$ has an associated counter and $\textsf{next}$ and
$\textsf{prev}$ pointers for maintaining doubly-linked lists. The pointers refer to indices of $\CSD$s.
The same pointers are stored across
all elements in a given $\CSD$ that belong to the same $\PD$. In the figure, the doubly linked list
for $\PD_b$ starts with $\CSD_k$, goes to $\CSD_0$ and ends at $\CSD_{\tf-1}$.}
\label{fig:SID}
\end{figure}

\paragraph{Functionality.}
Each insert, delete and query operation on an element $x$ is forwarded
to the $\CSD$ of index $\thash(x)$. Multiplicities are counted using
counters in the $\CSD$. If the counter becomes zero after a deletion,
the element is lazily removed from the $\CSD$.  The
pop operation is described separately in Sec.~\ref{sec:pop}.  Since
each $\CSD$ fits in a constant number of words, all these operations
require only a constant number of memory accesses.

\subsubsection{Pop Operations}\label{sec:pop}
To maintain Invariant~\ref{inv:SID} after $\del(x)$, a $\pop$
operation is issued to the $\SID$ of crate $\hc(x)$ whenever the
pocket dictionary of index $\hb(x)$ transitions from full to
non-full. The input is the index $\hb(x)$ of the pocket dictionary and
the output $(\hb(y), q(y),r(y)$ is the explicit part of an element $y$
stored in the $\SID$ with the property that $\hb(y) = \hb(x)$. A copy
of the element $y$ is then deleted from the $\SID$ and inserted into
the pocket dictionary of index $\hb(x)$.

\paragraph{Structure.}
In order to support $\pop$ operations, the $\SID$ keeps a doubly
linked list per $\PD_i$ (see Fig.~\ref{fig:SID}). The doubly linked list $L_i$ of $\PD_i$ contains
the indexes of the $\CSD$s that store elements $x$ such that
$\hb(x)=i$.
The doubly linked list $L_i$ is implemented by storing pointers to $\CSD$s
alongside the elements $\set{\hb(x),q(x),r(x))\mid \hb(x)=i}$ stored
in the $\CSD$s.

The fact that each crate has its own $\SID$ implies
that the pointers we store require only $\log \tf =O(\log w)$ bits. A
separate array is used to store  the list heads.

\paragraph{Functionality.}
Insertion of a new element $x$ to the $\SID$ is implemented as
follows.  Recall that $x$ is mapped to the $\CSD$ of index $\thash(x)$.
If $\CSD_{\thash(x)}$ already contains an element $y$ such that
$\hb(y)=\hb(x)$, then $x$ is inserted using the same pointers of $y$.
Otherwise, $x$ is inserted to the $\CSD$, its pointer is set to
the head of $L_{\hb(x)}$, and the head of the list is updated to $\thash(x)$.

If a duplicate element is inserted, the counter is merely
incremented. Deletion of $x$ from the $\SID$ decrements the counter.
If there is no other remaining element $y$ in $\CSD_{\thash(x)}$ with
$\hb(y)=\hb(x)$, then we update pointers in $L_{\hb(x)}$ to skip over
$\CSD_{\thash(x)}$.

A pop operation with input $i$ returns the first element $y$ from the
$\CSD$ pointed to by the head $L_i$ with $\hb(y)=i$. The counter of
$y$ is decremented. If $y$'s counter reaches zero, we continue as if
$y$ is deleted. The overhead in memory accesses required to
support the doubly linked lists is constant.

\subsubsection{Analysis: Proof of Lemma~\ref{lemma:crate-sid} in the dense case}
Even if the cardinality of the dynamic multiset
stored in the $\SID$ is at most $\tf$, there are two possible
causes for an overflow of the $\SID$:~(1) more than $\hf$ distinct
elements are stored in a $\CSD$ or,~(2) more than $\hatc$ copies of
the same element are stored.

Two technical challenges arise when we try to bound the probability of
overflow in the $\SID$:~(1) The triples $(\hb(x),q(x),r(x))$
stored in the $\SID$ are neither independent nor distributed
u.a.r. and, ~(2) The number of elements stored in the $\SID$ is
$\tf$, yet we seek an upper bound on the overflow probability that is
$1/\poly(n)$.

Consider a dynamic multiset stored in the $\SID$ over a polynomial
number of insertions and deletions to the Crate Dictionary.
\begin{claim}\label{crate-cpd-ovf}
If the cardinality of the dynamic multiset stored in the $\SID$ is at
most $\tf$, then the counting set dictionaries in the
$\SID$ do not overflow whp.
\end{claim}
\begin{proof}
Since we set $\hatc=\tf$, the maximum multiplicity of an element in
a $\CSD$ never exceeds the cardinality of the multiset in the
$\SID$. To complete the proof, we show that no more than $\hf$ distinct
elements are stored per $\CSD$.

Let $k$ be the number of distinct elements that are stored in the
$\SID$. The probability that more than $L$ distinct elements are
mapped to a specific $\CSD$ is bounded by
$\binom{k}{L}\cdot \tf^{-L} \leq \parentheses{\frac{e}{L}}^L$.
Setting $L= \hf \geq c\cdot w/O(\log(w))$ gives a probability that
is at most a polynomial fraction of $n$ (the exponent is linear in
$c$). The claim follows by applying a union bound over all the
$\CSD$s in a $\SID$ and all time steps.
\end{proof}
\begin{claim}\label{claim:sid mem}
If no $\CSD$ overflows, then every query, insert, delete, and pop to the
$\SID$ requires a constant number of memory accesses.
\end{claim}
\begin{proof}
Every insert, delete, or query operation is executed by accessing
the corresponding $\CSD$. The parametrization of the $\CSD$ implies
that it fits in a constant number of words (Claim~\ref{claim:CSD})
and therefore these operations require a constant number of memory
accesses. Maintaining the doubly linked lists accesses one entry in
the array and affects at most three different $\CSD$s: one $\CSD$
that stores the element, its predecessor, and its successor.
\end{proof}

\begin{claim}\label{claim:SID size}
The number of bits required for storing a $\SID$ is
$O(\tf w) + o(\CC)$.
\end{claim}
\begin{proof}
The $\SID$ has $\tf$ counting set dictionaries, each of which fits
in a constant number of words. Hence, the total space occupied by
the $\CSD$s is $O(\tf\cdot w)$ bits.
The array of list head pointers requires
$O\parentheses{\log w \cdot \CC/m}=o(\CC)$
bits because each pointer takes $\log (\tf)=O(\log w)$ bits.%
\footnote{A shared datastructure for all the elements that do not
fit in the ``bins'' appears in~\cite{arbitman2010backyard} and
also~\cite{bender2018bloom}). However, in our case, such a shared
$\SID$ would require $\Omega(n\log n)$ bits just for the array of
pointers. This would render the Crate Dictionary
space-inefficient.}
\end{proof}

\subsection{Analysis: Proof of Theorem~\ref{thm:crate-dict} in the dense case}
Consider a dynamic random multiset of cardinality at most $n$ specified by a
sequence of $\poly(n)$ insertions and deletions. We
bound the number of memory accesses of each operation:

\begin{claim}
If none of the components of the Crate Dictionary for the dense case overflow, then every insert, delete, and query operation requires a constant number of memory accesses.
\end{claim}
\begin{proof}
Each operation accesses one pocket dictionary, forwards the
operation to the $\SID$, or issues a pop operation to the $\SID$. By
Claim~\ref{claim:DictSmall size}, each pocket dictionary fits in a
constant number of words, and hence each operation on a pocket
dictionary requires only a constant number of memory
accesses. Operations on the $\SID$ require only a constant number
of memory accesses by Claim~\ref{claim:sid mem}.
\end{proof}

We now show that the Crate Dictionary in the dense case is
space-efficient. We set $\mu \triangleq \mumacro$. Since $m=w/\log(\rho)$,
 by the assumption that $w=\Omega(\log n)$  and given that 
we are in the dense case), we have that
$\mu= \sqrt{\frac{\ln(w)\log(\rho)}{w}}=o(1)$. We first observe that only a polylogarithmic fraction of
the bins in a crate will overflow whp.

\begin{claim}\label{claim:fullbins}
For every crate, at most $\frac{w^{2\beta/3}}{m}$ pocket
dictionaries are full in the crate whp.
\end{claim}
\begin{proof}
Fix a crate. Let $Z_i$ denote the number of elements that hash to
the pocket dictionary of index $i$ and
let $X_i$ be the indicator random variable that is $1$ if
$Z_i \geq (1+\mu)\cdot m$ (i.e. the pocket dictionary of index $i$
is full) and $0$ otherwise. The total number of full pocket
dictionary is denoted by $Y= \sum X_i$.

Since each element $x$ is drawn independently and u.a.r., we have
that $\hb(x)$ and $\hc(x)$ are values also chosen independently and
u.a.r. We get that $E[Z_i] \leq m$.  By Chernoff's bound, we have
that
$\Pr{Z_i \geq (1+\mu)\cdot m} \leq \exp\parentheses{-m\cdot
\mu^2/3}$. Indeed,
\begin{align*}
\Pr{Z_i \geq (1+\mu)\cdot m} &= \Pr{Z_i \geq (1+\mu) \frac{m}{E[Z_i]}\cdot E[Z_i]}\\ 
&\leq \Pr{Z_i \geq \parentheses{1+\mu \frac{m}{E[Z_i]}}\cdot E[Z_i]} \\
&\leq \exp\parentheses{-E[Z_i]\cdot \frac{\mu^2 m^2}{3 E[Z_i]^2}}  = \exp\parentheses{-m \mu^2/3 \cdot \frac{m}{ E[Z_i]}}  \\
&\leq \exp\parentheses{-m\mu^2/3}  \;.
\end{align*}

\medskip\noindent The expected number of full pocket dictionaries
satisfies:
\begin{align*}
E[Y] &\leq \frac{\CC}{m} \cdot \exp\parentheses{-m \mu^2/3} \;.
\end{align*}
Note that by our choice of $\mu$, we have that
\begin{align*}
6\cdot E[Y]&\leq  6\cdot \frac{\CC}{m} \cdot
\exp\parentheses{-m\mu^2/3} =\frac{ w^{2\beta/3}}{m} .
\end{align*}
The random variables $X_i$ are negatively associated, hence by Chernoff's bound~\cite{dubhashi1998balls},
\begin{align*}
\Pr{Y \geq \frac{w^{2\beta/3}}{m}} &\leq 2^{-w^{2\beta/3}/m}\;.
\end{align*}
For $\beta\geq\betaval$, we have $2\beta/3\geq 4+2\delta/3$, hence
$\Pr{Y \geq \frac{w^{2\beta/3}}{m}}$ is at most
$2^{-(3+\frac{2}{3}\delta)w}$.  Since $w=\Omega(\log n)$, the claim
follows.
\end{proof}
\medskip
\noindent
In the following claim we bound the maximum number of elements that
are assigned to the same pocket dictionary in a crate.
Recall the notation of the proof of Claim~\ref{claim:fullbins},
where $Z_i$ denotes the number of elements assigned to pocket
dictionary $i$.
\begin{claim}\label{claim:binocc}
For every $i$ and any random multiset set $\MM$ such
that $|\MM|\leq n$,
\begin{align*}
\Pr{Z_i\geq m \log n}\leq n^{-\omega(1)}\;.
\end{align*}
\end{claim}
\begin{proof}
The random variable $Z_i$ satisfies $E[Z_i] \leq m$. By
Chernoff's bound, we have that
\begin{align*}
\Pr{Z_i\geq m\log n}&\leq 2^{-m\log n} = n^{-m}\;,
\end{align*}
and the claim follows.
\end{proof}
\medskip\noindent Combining Claim~\ref{claim:fullbins} and
Claim~\ref{claim:binocc}, we get that:
\begin{claim}\label{claim:crate ovf}
The maximum number of elements stored in the $\SID$ of a crate is
$O(w^{2\beta/3} \cdot \log n)$ whp. Therefore, whp, no more than $\tf$
elements get stored in a specific $\SID$ if $\beta'>1+2\beta/3$.
\end{claim}

\medskip\noindent
We summarize the space requirements of the Crate Dictionary in the dense case.  
\begin{claim}\label{claim:dense space}
If $\beta-1>\beta'>1+2\beta/3$, then the number of bits required for
storing the Crate Dictionary in the dense case with parameters $n$
and $\rho$ is
\begin{align*}
(1+o(1))\cdot n\log(\rho) + (2+o(1)) n\;.
\end{align*}
\end{claim}
\begin{proof}
We have $n/m$ pocket dictionaries $\PD(m,(1+\mu)m,\log(1/\eps))$,
each of which takes
\begin{align*}
m + f(1+\log(\rho)) &= m\cdot \parentheses{ 2+\mu + (1+\mu)\log(\rho)}
\end{align*} bits (Claim~\ref{claim:DictSmall size}). In total, the pocket dictionaries  of the crates take $(1+\mu) n\log(\rho) + (2+\mu) n$ bits.
There are $n/C$ $\SID$'s. By Claim~\ref{claim:SID size}, all the $\SID$s 
require $O(\frac{n}{C}\cdot (\tf w + o(\CC))=o(n)$ bits.  The
claim follows.
\end{proof}
\medskip\noindent
Note that $\beta=\betaval$ and $\beta'=\betapval$ satisfy the premise
of the claim. 

\medskip\noindent
This completes the proof of Theorem~\ref{thm:crate-dict} in the dense case.

\subsection{The Case of Sets}
In this section we discuss modifications in the design and analysis of
a Crate Dictionary so that it supports sets rather than random
multisets. (Note that this case is not relevant for the design of a
filter based on a dictionary.)

One can simplify the design by replacing the counting set dictionaries
with set dictionaries (simply omit the counters). This modification
does not affect the number of memory accesses or the space efficiency.

The representation of the elements needs to be modified as follows.
Following~\cite{arbitman2010backyard}, we employ a random invertible
permutation $\tau:\UU\rightarrow\UU$. This transformation implies that
instead of maintaining the set $\DD$, we maintain the set
$\tau(\DD)$. Moreover, for every $\DD$, the set $\tau(\DD)$ is a
random subset of $\UU$ of cardinality $|\DD|$.  Invoking $\tau$
reduces the worst-case dataset to an average case one, as summarized
by the following assumption.
\begin{assumption}\label{assume:rnd}
By applying a random permutation, one may assume that the dataset
$\DD$ in the case of a dictionary for sets is a random subset of at
most $n$ elements from $\set{0,1}^u$.
\end{assumption}

Using Assumption~\ref{assume:rnd}, we can now proceed with the analysis as in the case of random multisets. The same guarantees follow since we obtain random hash values $\hc(x)$ and $\hb(x)$ that are uniformly distributed without replacements. The random variable $Z_i$ is a sum of negatively associated random variables and the standard Chernoff bounds apply~\cite{dubhashi1998balls}.

\section{The Crate Dictionary - Sparse Case}\label{sec:crate-dict-sparse}
In this section, we present a fully-dynamic dictionary for sets and random
multisets for the case that the relative sparseness $\rho$ satisfies
$w/\rho =o_n(1)$. We refer to this construction as the Crate
Dictionary for the sparse case. 

\subsection{Adaptive Remainders}\label{sec:adapt finger}

Recall that the element $x$ can be represented as the $4$-tuple
$(\hc(x),\hb(x),q(x),r(x))$, where each component is a disjoint block
of bits of $x$. The main challenge in the sparse case comes from
the fact that the $\PD$ remainder $r(x)$ can
no longer be stored by a pocket dictionary or might not even
fit in a word.

To overcome the fact that $r(x)$ is potentially too long, we compute a prefix
$\alpha(x)$ of $r(x)$, called the \emph{adaptive remainder} (see Fig.~\ref{fig:explicitbits}).  The
\emph{adaptive fingerprint} of $x$ is $(\hc(x),\hb(x),q(x),\alpha(x))$
and it satisfies :~(1)
distinct elements have distinct adaptive fingerprints that are prefix-free,
~(2) a pointer to the location of $r(x)$ is associated with the
adaptive fingerprint of $x$, and~(3) the adaptive remainder $\alpha(x)$ is short
enough, on average, so that it can be stored in a variable-length pocket
dictionary\footnote{Hence, $\alpha(x)$ is stored
twice: one time as an adaptive remainder and a second time as part
of $r(x)$.}.

To ensure that the adaptive fingerprints of distinct elements are
distinct and prefix-free, the adaptive remainders are variable-length;
we adaptively lengthen and shorten them (see
also~\cite{bender2018bloom}).  Specifically, we maintain the following
invariant for every subset $S\subseteq\DD$ of elements that either:
(1)~reside in the same $\PD$ and share a common $\PD$ quotient $q(x)$,
or (2)~reside in the same $\CSD_i$ and share the same $\PD$ index
$h^b(x)$ and $\PD$ quotient $q(x)$.
\begin{invariant}\label{inv:min}
For every $S\subseteq \DD$ defined above, the sum of the lengths of
the adaptive remainders is minimal subject to the constraint that
the set of adaptive fingerprints of distinct elements in $S$ is
prefix-free.\footnote{We note that the invariant does not exclude
the possibility that multiple copies of the same adaptive
fingerprint are stored. Two adaptive fingerprints are equal if and
only if they correspond to the same element.}
\end{invariant}

\begin{figure}[h]

\centering
\includegraphics[width=0.5\textwidth]{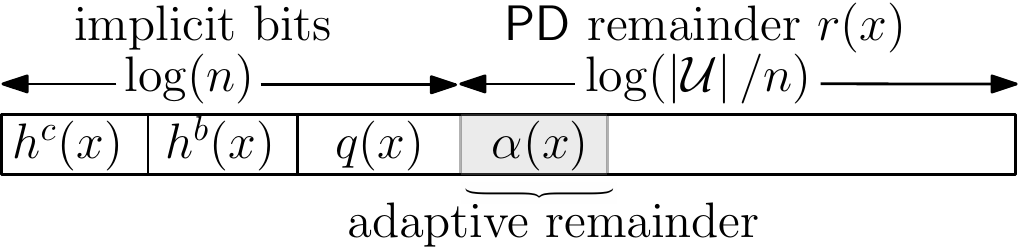}
\caption{\textbf{Adaptive Fingerprints.} The adaptive fingerprint for the Crate Dictionary is $(\hc(x),\hb(x),q(x),\alpha(x))$. The adaptive remainder $\alpha(x)$ is a prefix of $r(x)$.}

\label{fig:explicitbits}
\end{figure}

\paragraph{Maintaining Invariant~\ref{inv:min}.}

During insertions, adaptive remainders are minimally
extended until the adaptive fingerprints of
distinct elements are prefix-free.\footnote{Distinctness is always
achievable by taking all of the bits of an element.}  Consider the
following greedy procedure : when a new element $x$ is inserted, we
start with $\alpha(x)$ as the empty string and extend it as long as
the adaptive fingerprint of $x$ is a prefix of another adaptive
fingerprint of an element $y\neq x$. Additionally, we may also need to
extend the adaptive remainder of another element.  Since the set of
adaptive remainders before the insertion is prefix-free
(Invariant~\ref{inv:min}), at most one such adaptive fingerprint needs
to be extended. When duplicate elements are inserted, we store
duplicate adaptive remainders. See Appendix~\ref{sec:compute adaptive
remainders} for a detailed description of the procedure for
computing the adaptive remainders. We note that this procedure is
applied separately to subsets of fingerprints that are stored in the
same $\PD$ or $\CSD$ and that share the same $\PD$-quotient.  Upon deletions, adaptive remainders are
greedily shrunk while maintaining the prefix-freeness property.

\paragraph{Pointers and Pocket Motels.}
Since the adaptive remainder $\alpha(x)$ may be a proper prefix of
$r(x)$, we employ a separate data structure for maintaining a dynamic
multiset of $\PD$-remainders values. We refer to this data structure as a
\emph{pocket motel}. The pocket motel supports insertions, deletions,
and direct access via a pointer. See Appendix~\ref{sec:motel} for a
more detailed description of the implementation.

Let $\ptr(x)$ denote the pointer to the entry in the pocket motel that
stores $r(x)$. There is one pocket motel for each $\PD$ in the Crate
Dictionary (similarly, there is one pocket motel for each $\CSD$ in
the $\SID$). The cardinality of the multiset stored in the pocket
motel is equal to the cardinality of the multiset stored in the
corresponding $\PD$ or $\CSD$.  Since each $\PD$ and $\CSD$ is
responsible for less than $w$ elements, the length of $\ptr(x)$ is at
most $\log w$ bits.

\subsection{Structure}

The Crate Dictionary for the sparse case employs a variant of the
Crate Dictionary for the dense case to store the set of pointers
$\set{\ptr(x) \mid x\in\DD}$.  The dictionary is augmented to handle
adaptive remainders and pocket motels. The basic building block of the
Crate Dictionary for the sparse case consists of three synchronized
components: ~(1) $\PD$s and $\CSD$s that store pointers $\ptr(x)$,
~(2) variable-length $\PD$s and $\CSD$s for adaptive remainders
(Sec.~\ref{sec:var pocket dict}) and,~(3) pocket motels for storing
$\PD$-remainders $r(x)$. We describe how these components are designed
and how they are synchronized.

We follow the same parametrization that we use in the dense case. We
have $n/C$ crates, each consisting of $C/{m}$ pocket dictionaries
$\PD(m,f,\ell)$.  Each pocket dictionary is used for explicitly
storing the pointers $\ptr(x)$ and has its own pocket motel for
the $\PD$-remainders $r(x)$.

\paragraph{$\PD$-super-intervals.}\label{sec:super-interval structure.}
An alternative view of the structure is that the range $[n]$ is
partitioned into $n/m$ intervals such that one pocket dictionary is
assigned to each interval.  \emph{$\PD$-super-intervals} are obtained
by partitioning the set of intervals into groups of $\log w$
consecutive intervals.\footnote{For the sake of readability, we write
$\log w$ instead of $\ceil{\log w}$.} With each 
$\PD$-super-interval, we associate one $\VarPD$ for storing adaptive
remainders. See Fig.~\ref{fig:VarPD} for a depiction.

\begin{figure}[h]

\centering
\includegraphics[width=\textwidth]{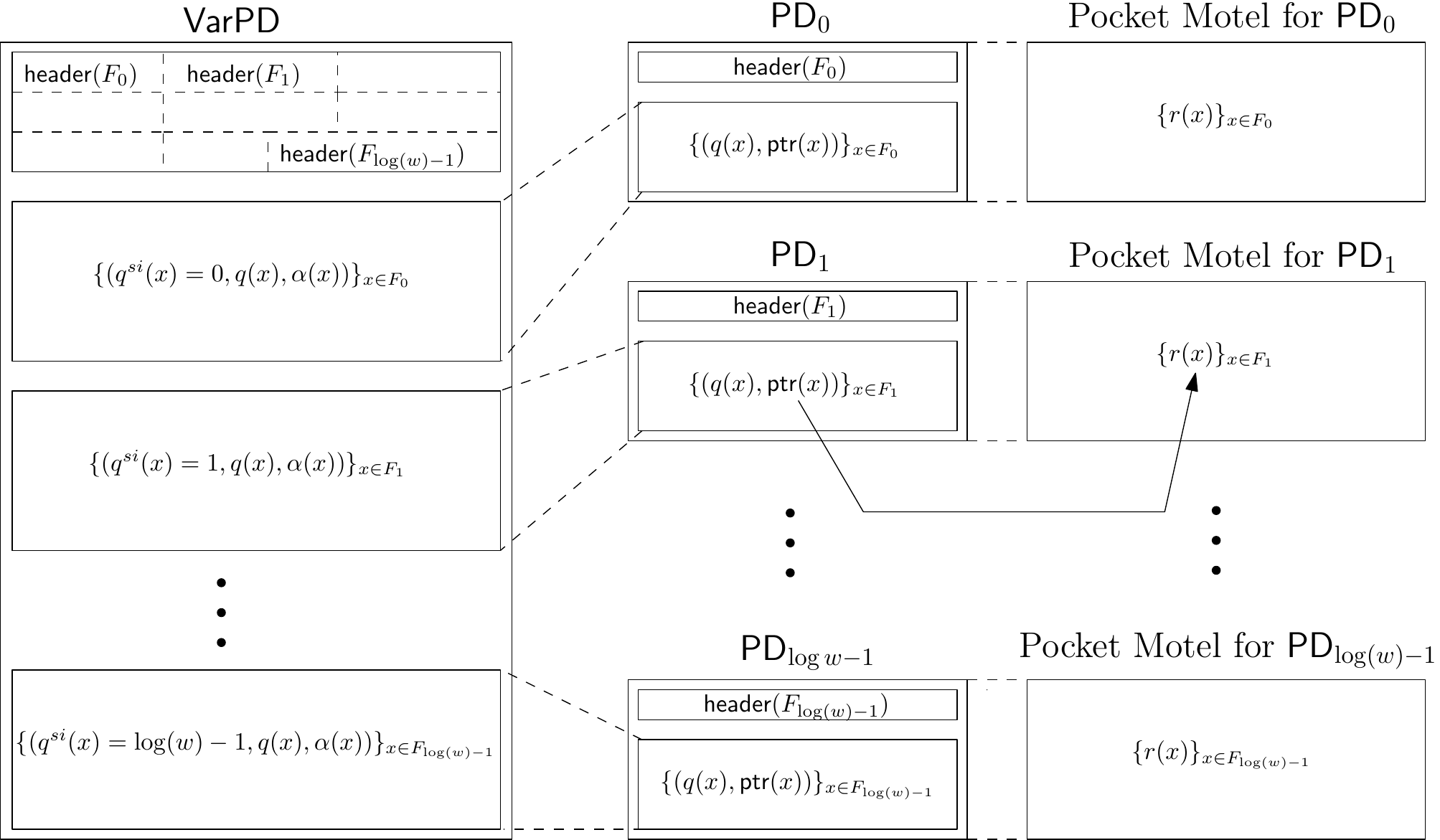}
\caption{\textbf{The components in a $\PD$-super-interval.} We denote
by $F_i$ the set of elements that correspond to $\PD_i$. Namely,
these are the elements $x$ with $\hb(x)=i$ that are not stored in
the $\SID$. The pocket dictionary $\PD_i$ stores $(q(x),\ptr(x))$ for
all $x\in F_i$, where $q(x)$ is stored implicitly. The $\VarPD$
stores $(\qsi(x), q(x), \alpha(x))$ for all
$x\in\set{F_i}_{i\in[\log w]}$, where $(\qsi(x),q(x))$ are stored
implicitly. The pocket model for $\PD_i$ stores $r(x)$ at location
$\ptr(x)$ in the pocket motel.}
\label{fig:VarPD}
\end{figure}

\paragraph{Storing pointers and $\PD$ remainders.}
The pocket dictionary of index $h^{b}(x)$ in crate $\hc(x)$ stores
pairs of the form $(q(x),\ptr(x))$. The order in which these pairs are
stored in a pocket dictionary is important.  The order is
non-decreasing lexicographic order according to $(q(x),r(x))$. (Note
the difference between the order and the stored pairs.) The value
$r(x)$ is stored in the pocket motel associated with the $\PD$ of
index $h^b(x)$, at the location within the pocket motel indicated by
$\ptr(x)$.

\begin{table}
\centering
\begin{tabular}{|c|c|l|}
\hline
symbol& value & role \\
\hline\hline 
$M$& $w$& the length of the super-interval \rule{0pt}{2.6ex}\\
$F$& $(1+o(1))\cdot M$& cardinality of the  multiset stored in a $\VarPD$  \\
$L$&  $O(F)$&  total length of adaptive remainders stored in a $\VarPD$\\
$\tilde{F}$& $ w$& cardinality of the set stored in a $\VarCSD$  \\
$\tilde{L}$&  $O(\tilde{F})$& total length of adaptive remainders stored in a $\VarCSD$\\
$\qsi(x)$& $\hb(x) \pmod{\log w}$ &
index of the $\PD$ within the $\PD$-super-interval\\
\hline
\end{tabular}
\caption{Parametrization of $\VarPD$ and $\VarCSD$}
\label{tbl:varpd parameters}
\end{table}

\paragraph{Storing adaptive remainders.}
For each $\PD$-super-interval, we allocate a variable-length pocket
dictionary $\VarPD(M,F,L)$. The parameters for the $\VarPD$ are listed
in Table~\ref{tbl:varpd parameters}. Note that the quotient of an element
inside its $\VarPD$ is given by the pair $(\qsi(x),q(x))$, where $\qsi(x) = \hb(x) \pmod{\log w}$.
The elements within a $\VarPD$ are stored in non-decreasing lexicographic
order of $(\qsi(x),q(x),r(x))$. For more details on what is stored in each
component of the $\PD$-super-interval, see Fig.~\ref{fig:VarPD}.

\paragraph{Synchronization in the $\PD$-super-interval.}
The fact that the pointers in the $\PD$s and the adaptive remainders
in the $\VarPD$s are stored in the same order allows for
synchronization. 
\begin{claim}\label{claim:sync PD}

The position of $\ptr(x)$ in the $\PD$ is determined by the position
of $\alpha(x)$ in the $\VarPD$. 
\end{claim}
\begin{proof}
Since adaptive remainders in the $\VarPD$ are stored in the order
$(\qsi(x),q(x),r(x))$,  adaptive remainders from the same $\PD$
map to a contiguous block within the body of the $\VarPD$. Within
that block of the $\VarPD$, they are stored in the order
$(q(x),r(x))$, which is the same order that $\ptr(x)$ is stored in
the $\PD$.
\end{proof}

\paragraph{$\CSD$-super-intervals.}

\begin{figure}[h]
\centering
\includegraphics[width=\textwidth]{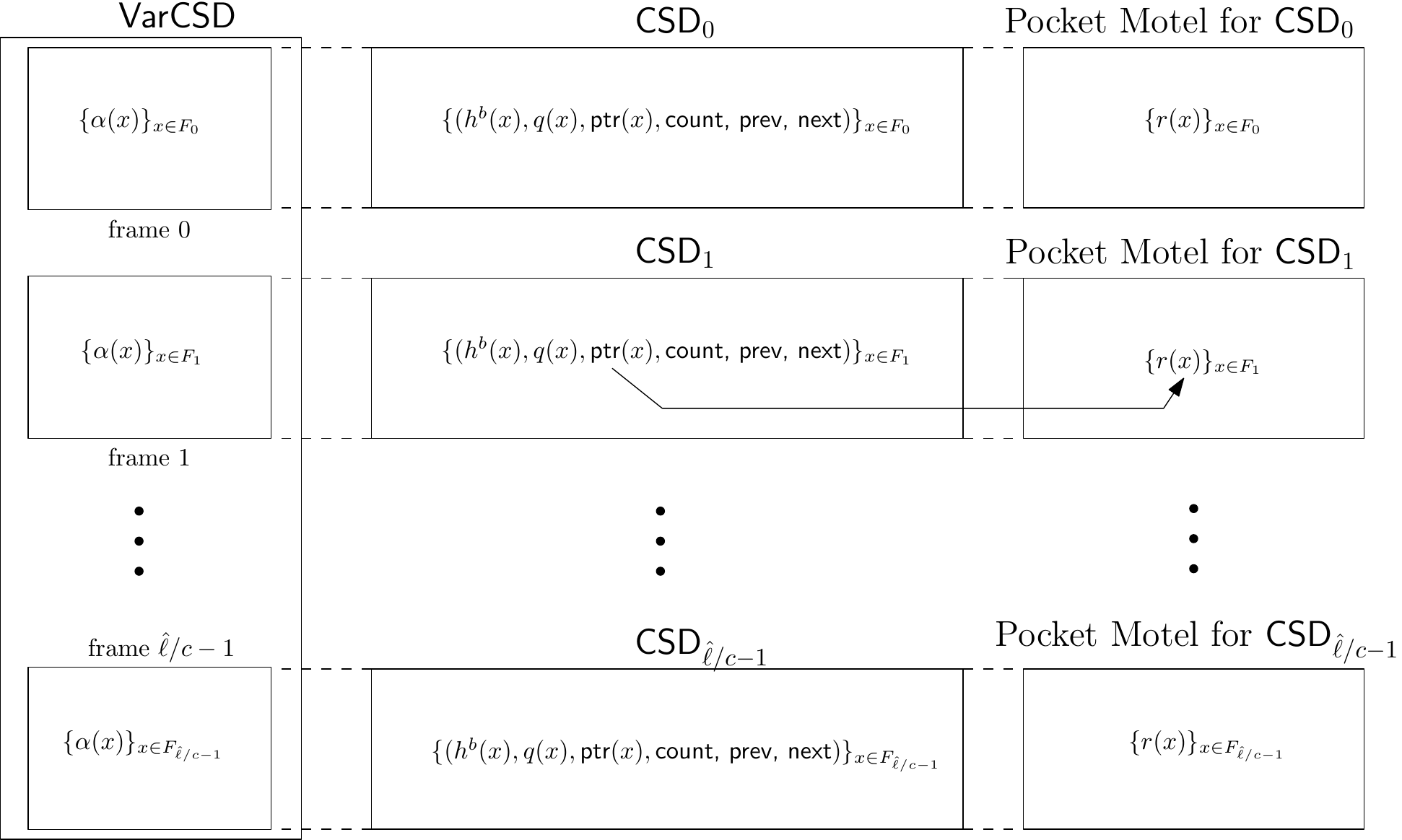}
\caption{\textbf{The components in a $\CSD$-super-interval.}
In this figure we depict the organization related to the
first $\CSD$-super-interval. We denote by $F_i$ the set of
elements that correspond to $\CSD_i$. Namely, these are the
elements $x$ with $\thash(x)=i$ that are not stored in the
$\PD$s. Each $\CSD_i$ explicitly stores
$(\hb(x),q(x),\ptr(x))$ alongside a multiplicity counter and doubly
linked list pointers. The order is non-decreasing
lexicographic order of $(\hb(x),q(x),r(x))$. The $\VarCSD$
stores $\alpha(x)$ for all $x\in\set{F_i}_{i\in[\hl/c]}$ in
non-decreasing lexicographic order of
$(\thash(x),\hb(x),q(x),r(x))$. The pocket model for
$\CSD_i$ stores $r(x)$ at location $\ptr(x)$.}
\label{fig:VarCSD}
\end{figure}

The $\SID$ mimics the structure of synchronized triples of pocket
dictionaries, pocket motels, and variable-length pocket
dictionaries (See Fig.~\ref{fig:VarCSD}).
We  group $\hl/c$
consecutive $\CSD$s and pocket motels into a \emph{$\CSD$-super-interval}. 
In each distributed $\SID$, there are $\tf$ $\CSD$s, each storing a
multiset whose support is of cardinality at most $\hf = cw/\hl$. We
attach a pocket motel to each $\CSD$.

Each $\CSD$-super-interval is assigned a variable-length counting set
$\VarCSD(\tilde{F},\tilde{L})$ for storing adaptive remainders. The $\VarCSD$ 
stores a set of elements since counters are already maintained in the $\CSD$.
The data structure is parametrized by the cardinality of the set and the total length
of the adaptive remainders. More details on the $\VarCSD$ can be found in Appendix~\ref{sec:var count set dict}.

Each $\CSD$ stores triples $(\hb(x),q(x),\ptr(x))$, which are stored
in non-decreasing lexicographic order of $(\hb(x), q(x),r(x))$. The
pocket motel associated with the $\CSD$ stores $r(x)$ at location
$\ptr(x)$. Each $\VarCSD$  stores
$\alpha(x)$ in non-decreasing lexicographic order of
$(\thash(x), \hb(x), q(x),r(x))$ . To achieve this, the $\VarCSD$ 
stores adaptive remainders from the same $\CSD$ in one contigous block called a \emph{frame}.
We have that the $\CSD$ and the
$\VarCSD$ are synchronized. 

\begin{claim}\label{claim:sync SID}
One can deduce the position of $\alpha(x)$ in the $\VarCSD$ from the
position of $(\hb(x),q(x),\ptr(x))$ in the $\CSD$.
\end{claim}
\begin{proof}
All the adaptive remainders from elements assigned to the same $\CSD$
map to one contiguous block of entries in the $\VarCSD$ (i.e. a frame). 
Within the same $\CSD$, we have that the set of triples $(\hb(y),q(y),\ptr(y))$ such that 
$(\hb(y),q(y)) = (\hb(x), q(x))$ is stored in a block of contiguous entries.
This block corresponds to a similar block in the $\VarCSD$ because the adaptive remainders in the same
frame of the $\VarCSD$ respect the same lexicographic order. In other words, the relative
position of $\alpha(x)$ in the frame is the same as the relative position of $\ptr(x)$ within the $\CSD$.

\end{proof}

\subsection{Functionality}

One important property of our construction is that it allows us to
find the ``location'' of an element $x$ (i.e. $\ptr(x)$) using a constant number of
memory accesses, regardless of the length of $x$.
Consider an element $x\in \UU$, an element $y_1\in \DD$ stored in the
$\PD$, and an element $y_2$ stored in the $\CSD$. Note that,
in the following definition, the conditions for $y_i$ to match $x$ depend on whether
$y_i$ is stored in a $\PD$ or in a $\CSD$.
\begin{definition}
We say that $y_1$ \emph{matches} $x$ if
$(\hc(y_1),\hb(y_1),q(y_1)) = (\hc(x),\hb(x),q(x))$, and
$\alpha(y_1)\subseteq r(x)$.  We say that $y_2$ \emph{matches} $x$ if
$(\thash(y_2), \hc(y_2),\hb(y_2),q(y_2)) = (\thash(x),
\hc(x),\hb(x),q(x))$, and $\alpha(y_2)\subseteq r(x)$.  We say that
$y$ \emph{fully-matches} $x$ if $y$ matches $x$ and $r(y)=r(x)$.
\end{definition}

\medskip\noindent
The implementation of $\query(x)$ is based on the following
claim.
\begin{claim}\label{claim:ptr}
For every $x\in \UU$, within a constant number of memory accesses,
one can find at most two pointers ($\ptr(y_1)$ in a $\PD$ and
$\ptr(y_2)$ in a $\CSD$) that point to matches with $x$. Moreover,
$x\in\DD$ if and only if at least one of the matches is also a full
match.
\end{claim}
\begin{proof}
Given an element $x$, we look for it in the $\PD$ and potentially
also in the $\CSD_{\thash(x)}$. This search gives us a set of
candidate elements
$S = \set{y \mid (\hc(x),\hb(x),q(x)) = (\hc(y),\hb(y),q(y))}$.  We
then locate the corresponding block of adaptive remainders
$\set{\alpha(y) \mid y\in S}$ in a constant number of memory
accesses (Claim~\ref{claim:sync PD} and Claim~\ref{claim:sync
SID}). If none of the adaptive remainders match, then we respond
that $x$ is not in the dataset. Otherwise, we have that at most one
match within the $\PD$ and one match within the $\CSD_{\thash(x)}$
(Invariant~\ref{inv:min}). Both the $\PD$, $\CSD$ and their variable
length counterparts fit in a constant number of words. Hence finding
a match requires only a constant number of memory accesses.  A full match
means $x$ is in $\DD$, otherwise $x$ is not in $\DD$.
\end{proof}

We discuss separately how insertions and
deletions are managed in the $\PD$-super-intervals and in the
$\CSD$-super-intervals.

\paragraph{Insertions and deletions in the $\PD$-super-interval.}
An $\ins(x)$ first checks whether the $\PD_{h_b(x)}$ in crate $\hc(x)$
is full.  If it is not full, the following steps are performed:
(1)~Insert $r(x)$ into the pocket motel, get $\ptr(x)$ and insert
$(q(x),\ptr(x))$ to $\PD_{\hb(x)}$. (2)~Locate the block of entries in
the $\VarPD$ that corresponds to $(\hb(x),q(x))$.  Compute the
adaptive remainder $\alpha(x)$ and insert $(\bh(x), q(x),\alpha(x))$
to $\VarPD_{\floor{\hb(x)/\log w}}$ (this may require extending an
existing adaptive remainder of another element). If $\PD_{\hb(x)}$ is full, 
then the insertion is forwarded
to the $\SID$.

A $\del(x)$ is processed by first locating $x$ by a $\query(x)$
operation. If the adaptive fingerprint of $x$ is found in the $\PD$,
deletion is similar to the dense case with the additional steps of
deleting the element's $\PD$-remainder from the pocket motel and its
adaptive remainder from the $\VarPD$ (which might require shortening
an adaptive remainder of another element in the $\VarPD$). If, in
addition, the $\PD$ was full prior to deletion, then a $\pop(\hb(x))$
operation must follow to maintain Invariant~\ref{inv:SID}. The element
returned by the $\pop$ is then inserted into $\PD_{h_b(x)}$.

\paragraph{Insertions and deletions in the $\CSD$-super-interval.}
An $\ins(x)$ to the $\SID$ first checks if $x$ is already stored in
the corresponding $\CSD$ interval. If $x$ is already stored, then it
increments its counter in the $\CSD$ and returns. If $x$ is a new element, then
$r(x)$ is first inserted into the pocket motel for $\CSD_{\thash(x)}$
and $(\hb(x),q(x),\ptr(x))$ is inserted into the
$\CSD_{\thash(x)}$. The $\VarCSD$ frame corresponding to $\CSD_{\thash(x)}$
is read and the adaptive remainder $\alpha(x)$ is
calculated. Finally, $\alpha(x)$ is inserted in the appropriate position in the
$\VarCSD$ based on the location of $\ptr(x)$ in the $\CSD$.
The linked lists are maintained as in the dense case.

A $\del(x)$ operation proceeds like a $\query(x)$ and updates the
counters and the linked list as in the dense case. In addition, if the counter reaches zero,
the element's $\PD$-remainder must be freed from the pocket dictionary and its adaptive 
remainder removed from the $\VarCSD$.  This might lead to another 
adaptive remainder in the $\VarCSD$ requiring shortening.

A $\pop(\hb(x))$ operation accesses the $\CSD$ pointed to by the head of the
linked list, deletes an element that belongs to $\PD_{\hb(x)}$, and inserts this element to $\PD_{\hb(x)}$.

\subsection{Analysis: Proof of Theorem~\ref{thm:crate-dict} in the sparse case}

We begin by proving that an overflow does not occur whp.
The probability of overflow of the $\PD$s, $\CSD$s and their pocket motels is summarized in the following claim.

\begin{claim}\label{claim:pocket-ovf-sparse}
An overflow of a pocket dictionary or a counting set dictionary does
not occur whp.
\end{claim}
\begin{proofsketch}
We employ the same load balancing arguments that we used to bound
the probability of overflow for the dictionaries in the dense case
in which the relative sparseness equals $w$. Indeed, the $\PD$s and
$\CSD$s now store $\ptr(x)$ (instead of $r(x)$), where the length of
$\ptr(x)$ is at most $\log w$. 
\end{proofsketch}
\medskip\noindent
Note that a pocket motel overflows if and only if its corresponding
$\PD$ (or $\CSD$) overflows.

\begin{claim}\label{claim:varpd-ovf}
The variable-length dictionaries $\VarPD(M,F,L)$ and
$\VarCSD(\tilde{F},\tilde{L})$ do not overflow whp.
\end{claim}
\begin{proofsketch}
If none of the $\PD$s overflow, then each $\VarPD$ and stores at
most $(1+o(1))\cdot w$ adaptive remainders.  Similarly, if none of
the $\CSD$s overflow, then each $\VarCSD$ stores at most $w$
distinct adaptive remainders.  Hence, with high probability, the
$\VarPD$s and $\VarCSD$s also do not overflow due to too many
elements.  The average length of $\Omega(\log n)$ adaptive remainders in the
$\VarPD$s and $\VarCSD$s is constant whp~\cite[Lemma
$15$]{bender2017bloom}.
Hence, an overflow does not occur whp, as required.
\end{proofsketch}
\medskip\noindent
Combining Claim~\ref{claim:pocket-ovf-sparse} with Claim~\ref{claim:varpd-ovf} gives us:

\begin{claim}
For every $\poly(n)$ sequence of insertions, deletions and queries, the
components of the Crate Dictionary for the sparse case do not overflow whp.
\end{claim}

We proceed by bounding the number of memory accesses per operation.
Recall that the
word length satisfies $w=\Omega(\log n)$.

\begin{claim}
If none of the components of the Crate Dictionary for the sparse
case overflow, then every query and insert operation requires at most
$2\cdot\frac{\log (|\UU|/n)}{w}+O(1)$ memory accesses, and every 
delete operation requires at most
$4\cdot \frac{\log (|\UU|/n)}{w}+O(1)$ memory accesses.
\end{claim}
\begin{proof}
For $\query(x)$, Claim~\ref{claim:ptr} states that
locating $\ptr(y)$ with the desired properties can be done in a
constant number of memory accesses. Reading $r(y)$ requires an
additional $\frac{\log (|\UU|/n)}{w}$ memory accesses because $r(y)$
is $\log (|\UU|/n)$ bits long. Since we may need to follow two
pointers, one from the $\PD$ and the second from the $\CSD$, the
total number of memory accesses is at most
$2\cdot\frac{\log (|\UU|/n)}{w}+O(1)$.

An insert operation $\ins(x)$ attempts an insert in the $\PD$. If
the $\PD$ is not full, we consider two cases.  (1)~A previous copy
of $x$ is not already in the $\PD$. In this case, we insert $x$ and
may need to extend an adaptive remainder of another element. This
requires $2\cdot\frac{\log (|\UU|/n)}{w}+O(1)$ memory accesses.
(2)~Previous copies of $x$ are already in the $\PD$. In this case,
they all have the same adaptive remainder. We need
$\frac{\log (|\UU|/n)}{w}$ memory accesses to compare the
$\PD$-remainder, and also need to insert an additional copy of the
$\PD$-remainder of $x$ into the pocket motel.  If the $\PD$ is full,
then we proceed with an insertion to the $\CSD$, which also takes at
most $2\cdot\frac{\log (|\UU|/n)}{w}+O(1)$ memory accesses.

The worst case for a delete operation $\del(x)$ is when the $\PD$ is
full and $x$ is stored in the $\PD$. In this case, we need to:
(1)~find $x$ (as in a query in a $\PD$) and delete it, (2)~pop an
element from the $\SID$ (which requires reading its
$\PD$-remainder), and (3)~insert the popped element to the
$\PD$. Note that shortening the adaptive remainder of $y$ does not
require reading $r(y)$. In total, the delete operation, in this case,
requires at most $4\cdot \frac{\log (|\UU|/n)}{w}+O(1)$ memory accesses.
\end{proof}

\medskip\noindent
The space requirement of the Crate Dictionary is summarized in the
following claim.
\begin{claim}
The Crate Dictionary for the sparse case requires
$(1+o(1))\cdot n \log \rho$ bits whp.
\end{claim}
\begin{proof}
The $\PD$s and $\CSD$s used for storing pointers are parametrized
as in the Crate Dictionary with relative sparseness equaling $w$. Therefore,
they require $(1+o(1))\cdot n  \log w + O(n)$ bits (Claim~\ref{claim:dense space}).
The variable-length pocket
dictionaries require $O(n)$ bits whp (Claim~\ref{claim:varpd-ovf}).  Per element, the length of
the $\PD$ remainder is $\log(|\UU|/ n)$.  Hence, the total space
consumed by the pocket motels is
$(1+o(1))\cdot n\log(\size{\UU}/ n) + o(n\log w)$ (Claim~\ref{claim:pocket-motel}). Since
$|\UU|/n=\rho$ and $w=o(\rho)$ the required
bound follows.
\end{proof}
\medskip\noindent
This completes the proof of Theorem~\ref{thm:crate-dict} in the sparse case.

\section{Application: Static Retrieval}\label{sec:retrieval}
In the static retrieval problem, we are given a (static) set $\DD\subseteq \UU$ and a function
$f:\DD\rightarrow \set{0,1}^k$ (one can think of $f(x)$ as the label of $x$). We would like to support queries of the form:
\begin{align*}
\query(x)&\triangleq
\begin{cases}
f(x) & \text{if $x\in \DD$}\\
\text{arbitrary value in $\set{0,1}^k$}&\text{if $x\notin\DD$.}
\end{cases}
\end{align*}

The retrieval problem has mainly been discussed in the static case,
although fully-dynamic constructions also
exist~\cite{demaine2006dictionariis,mortensen2005dynamic}. Several constructions rely on
perfect hashing
techniques~\cite{botelho2007simple,botelho2013practical,
genuzio2016fast, hagerup2001efficient}. The data structures obtained
this way require $pk+cn$ number of bits, where $p$ is the size of the
range of the perfect hash function and $(c+o(1))n$ is the number of bits
the perfect function occupies. For these constructions, either $p/n$
is a constant strictly greater than $1$ ($\approx 1.23$) or the bounds
depend on the size of the universe (which can be arbitrarily
large). A different line of approach is
based on representing the keys $x$ as binary rows in a system of
linear equations~\cite{DBLP:conf/stacs/DietzfelbingerW19,
chazelle2004bloomier,dietzfelbinger2008succinct}. These
constructions generally assume that $k$ is constant. We refer to the
paper of Dietzfelbinger and
Walzer~\cite{DBLP:conf/stacs/DietzfelbingerW19} for a thorough
comparison of the different data structure parameters.

\subsection{Sketch of the Proof of Theorem~\ref{thm:retrieval}}
The design of the retrieval data structure is very similar to
that of the Crate Dictionary. We briefly describe the modifications. For every $x\in\DD$,
we compute its adaptive remainder $\alpha(x)$. This adaptive remainder
will be stored alongside $f(x)$ in the following sense.

When $k = O(\log \log n)$, we use a variant of the Crate Dictionary
for the sparse case in which $f(x)$ is stored instead of $\ptr(x)$ in
the $\PD$ or the $\CSD$.
The location of $f(x)$ in the data structure is computed in the same way
in which the location of $\ptr(x)$ is computed in the dictionary. We omit
pocket motels since we no longer need to store $r(x)$.  Since the total length of the adaptive remainders is $O(n)$ with high
probability, we get that the total space occupied by this construction
is $(1+o(1))nk+O(n)$.

When $r=\omega(\log \log n)$, we use the construction for the dense
case directly and store $f(x)$ in the pocket motel instead of
$r(x)$. The total space occupied by this construction is
$(1+o(1))nk$.

To construct the data structure, we can insert the elements 
incrementally. Since each such inserts takes $O(1)$, we get that the runtime
of inserting the whole set is $O(n)$. With very low probability, the data structure
might overflow. In that case, we pick a new hash function and redo the process.
This leads to an overall construction runtime of $O(n)$ whp. 

\section*{Acknowledgments}
The authors would like to thank Michael Bender, Martin Farach-Colton,
and Rob Johnson for introducing this topic to us and for interesting
conversations.

\bibliography{main}

\appendix
\section{Pocket Motels}\label{sec:motel}
The pocket motel is a data structure that supports insertions, deletions, and
direct access via a pointer. It stores $\PD$-remainders of fixed
length.  An $\ins(r(x))$ operation returns a pointer $\ptr(x)$ to the entry
where $r(x)$ is stored. A $\del(\ptr(x))$ frees the entry. The pointer
$\ptr(x)$ can be used to directly access the element.

The pocket motel can be implemented using an array. See Fig.~\ref{fig:motel}
for a depiction. Each entry in the
array stores either the $\PD$ remainder $r(x)$ (i.e. it is
\emph{occupied}) or a pointer to the next free entry in the array
(i.e. it is \emph{vacant}). A separate $\textsf{head}$ pointer is
maintained for the beginning of the linked list of vacant entries.

\begin{figure}[h]
\begin{center}
\includegraphics[width=0.3\textwidth]{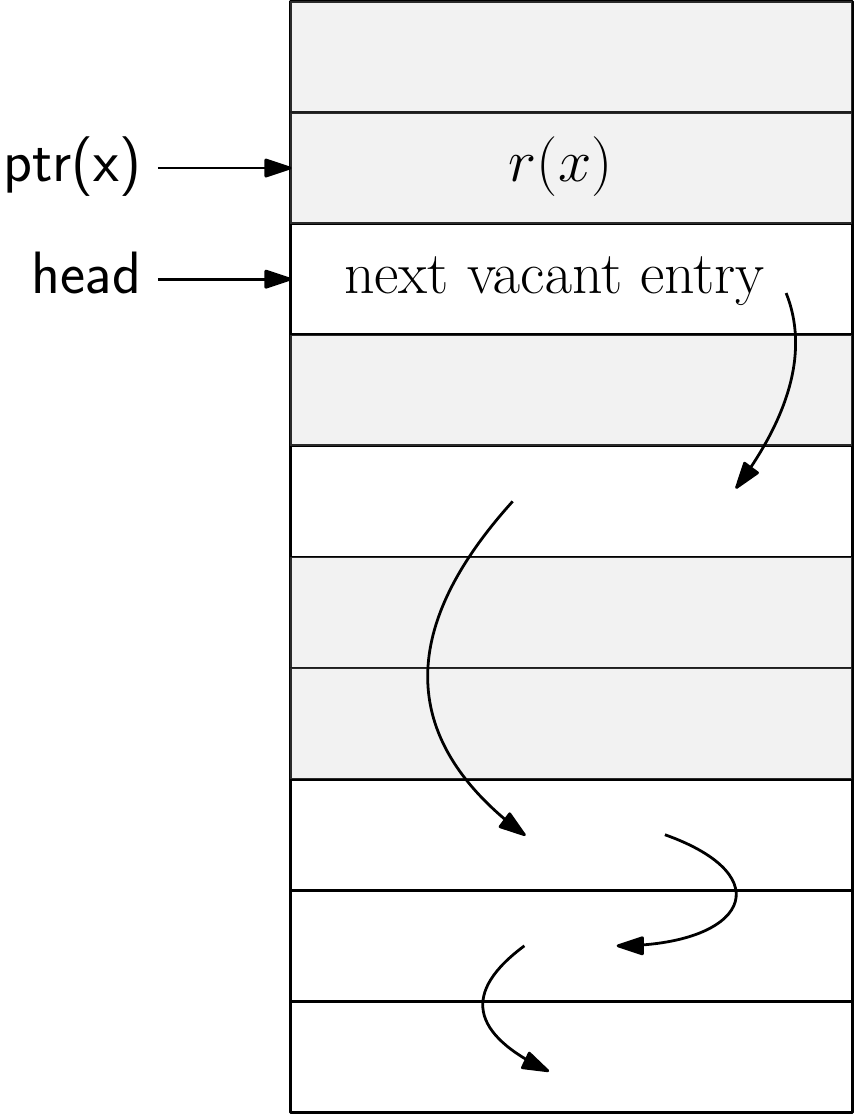}
\end{center}
\caption{\textbf{The pocket motel.} Occupied entries are shown in grey, vacant
entries are in white and form a linked list of vacant entries.}
\label{fig:motel}
\end{figure}

Insertion to the pocket motel is implemented by finding the first
vacant entry in the list (pointed to by the \textsf{head}), updating
the \textsf{head} pointer and writing the $\PD$ remainder of the
element in the array at that entry. Deletion is implemented by adding
the entry to the linked list of vacant entries. We get that:
\begin{claim}\label{claim:pocket-motel}
A pocket motel storing the $\PD$ remainders of $k \leq |U|/n$
elements requires
$k\cdot\parentheses{1+\log(\size{\UU}/n)} + \log k$ bits of space.
\end{claim}
\begin{proof}
We need a pointer to point to the first vacant entry in the motel.
Each entry is $\ceil{\log(\size{U}/n)}$ bits long.
\end{proof}

\section{$\SID$ Building Blocks}\label{sec:aux}

In this section, we describe auxiliary data structures ($\CSD$ and
$\VarCSD$) that we use in the design of the $\SID$.  The $\CSD$ has a
role that is similar to that of a $\PD$, and the $\VarCSD$ has a role
that is similar to that of a $\VarPD$.  The design of $\CSD$s and
$\VarCSD$s need not be space-efficient because the $\SID$ stores a
sublinear number of elements.

\subsection{Counting Set Dictionaries ($\CSD$s)}\label{sec:count set
dict}
The counting set dictionary ($\CSD$) is a dictionary for
small general multisets.  The parameters of a $\CSD$ are: $\hf$~-~an upper bound on
the cardinality of the support of the stored multiset, $\hatc$~-~an
upper bound on the maximum multiplicity of elements in the multiset,
$\hl$~-~the length of the elements in bits. Counting set dictionaries
are used in the implementation of the $\SID$ of each crate in the Crate
Dictionary.

We denote the data structure storing such a multiset by
$\CSD(\hf,\hl,\hatc)$. We use a brute force implementation in
which elements with their multiplicity (using $\ceil{\log \hatc}$
bits) are packed contiguously. Elements are stored in non-decreasing lexicographic order.

\begin{claim}\label{claim:CSD} The number of bits that
$\CSD(\hf,\hl,\hatc)$ requires is at most
$\hf(\hl+\log\hatc+1)$.  If  $\max \set{\hl,\log \hatc}<w$, and 
$\hf=O\parentheses{w/\max \set{\hl,\log \hatc}}$, then
$\CSD(\hf,\hl,\hatc)$ fits in a constant number of
words.
\end{claim}

We say that a $\CSD(\hf,\hl,\hatc)$ overflows whenever
the cardinality of the support of the stored multiset exceeds $\hf$
or the maximum multiplicity exceeds $\hatc$.

\subsection{Variable-Length Counting Set Dictionaries
($\VarCSD$s)}\label{sec:var count set dict}
The variable-length counting set dictionary ($\VarCSD$) is a
dictionary that stores sets of adaptive remainders of elements from a
$\CSD$-super-interval. The parameters of a $\VarCSD$ are: $\tilde{F}$-
an upper bound on the cardinality of the stored set and $\tilde{L}$-
an upper bound on the total length of the elements of the set. We note
this data structure by $\VarCSD(\tilde{F},\tilde{L})$.

The $\VarCSD$ is maintained synchronized with the $\CSD$s in the same $\CSD$-super-interval. Specifically, all the adaptive remainders of elements stored in the same $\CSD$ are stored in one contiguous block of entries in the corresponding $\VarCSD$. We call such a block a \emph{frame}. Within a frame, the $\VarCSD$ stores $\alpha(x)$ in non-decreasing lexicographic order of $(\hb(x),q(x), r(x))$. To enforce this, at insertion time, $\alpha(x)$ is stored in the same relative position within a frame as $(\hb(x),q(x),\ptr(x))$ is stored in the $\CSD$.

The $\VarCSD$ uses an alphabet of four symbols that includes two types
of separators: one separator distinguishes between successive frames
and the other separator distinguishes between the adaptive remainders
in the same frame. Since every $\CSD$-super-interval contains $\hl/c$
$\CSD$s, there are $\hl/c$ frame separators in each $\VarCSD$ (i.e. we
also record empty frames).

\begin{claim}\label{claim:VarCSD} The number of bits that
$\VarCSD(\tilde{F},\tilde{L})$ requires is at most
$2 (\tilde{F}+\tilde{L} + \hl/c)$.  If  $\max \set{\tilde{F},\tilde{L}}<w$, then
$\VarCSD(\tilde{F},\tilde{L})$ fits in a constant number of
words.
\end{claim}
\begin{proof}
We require $2\tilde{F}$ bits for the separators between the adaptive
remainders. The adaptive remainders themselves take at most
$2\tilde{L}$ bits. Finally, the frame separators require $2\hl/c$
bits, since there are $\hl/c$ frames in a $\VarCSD$. Since $\hl=O(\log w)$, the claim follows. 
\end{proof}

We say that a $\VarCSD(\tilde{F},\tilde{L})$ overflows whenever
the cardinality of the stored set exceeds $\tilde{F}$
or the total length of the adaptive remainders exceeds $\tilde{L}$. 

\section{Computation of Adaptive Remainders}\label{sec:compute adaptive remainders}
In this section, we describe a greedy process for computing the
adaptive remainders that satisfy Invariant~\ref{inv:min}. We are given
as input $f$ pairs of the form $(q_i, r_i)$, where
$q_i\in \set{0,1}^{\log w}$ and $r_i\in\set{0,1}^*$.  The output is a
set of adaptive remainders $\set{\alpha_i}_{i\in [f]}$ such that:
(1)~$\alpha_i$ is a prefix of $r_i$, (2)~the set
$\set{(q_i,\alpha_i)}_{i\in [f]}$ is prefix-free, and (3)~the sum of lengths
$\sum_{i\in[f]} \size{\alpha_i}$ is minimal.

The pairs are processed sequentially. Suppose we are processing
$(q_i,r_i)$ for some $i\in[f]$. We can view the pairs
$\set{(q_j,\alpha_j)}_{j<i}$ as if they were stored in a trie.
Interior nodes of the trie have at most two children, and every leaf
corresponds to a distinct pair $(q_j,\alpha_j)$. The adaptive remainder
$\alpha_i$ is then computed in a greedy fashion.  We begin with
$\alpha_i = \Lambda$ (i.e. the empty string) and extend $\alpha_i$ (by
``copying'' bits of $r_i$) in the following way:
\begin{enumerate}
\item If $q_i$ corresponds to an interior trie node, then we extend
$\alpha_i$ until we create a new leaf in the trie or reach an old
leaf. 
\item If $(q_i,\alpha_i)$ reaches an old leaf $(q_j,\alpha_j)$ for
some $j<i$, then we extend both $\alpha_i$ and $\alpha_j$ until
$\alpha_i\neq \alpha_j$.
\item Once $(q_i,\alpha_i)$ reaches a new leaf in the trie, no further
extension is required.
\end{enumerate}

We note that, if the pairs $\set{(q_i,\alpha_i)}_{i\in [f]}$ fit in a constant
number of words, then computing each $\alpha_i$ can be done in a
constant number of operations. See~\cite{bender2017bloom} for details.

\end{document}